\documentclass[a4paper,11pt]{article}

\usepackage[latin1]{inputenc}
\usepackage[T1]{fontenc}
\usepackage{lmodern}
\usepackage{braket}%
\usepackage{booktabs}%
\usepackage{multirow}%
\usepackage{float}%
\usepackage{subfig}
\usepackage{color}
\usepackage{graphicx,tikz} %\definecolor{grey}{rgb}{.7,.7,.7}
\usepackage{epstopdf} %\epstopdfsetup{update} % only regenerate pdf files when eps file is newer

\usepackage{natbib}
\usepackage{setspace}%\singlespacing%\doublespacing%\setstretch{1.1}
\renewcommand{\theenumi}{\roman{enumi}}
\usepackage{enumitem}\setlist[enumerate]{label={\rm(\theenumi)}}

\newcommand{\comment}[1]{}

\usepackage{amsmath,amssymb,amsthm}
\usepackage{mathrsfs} % provide \mathscr, alternative to \mathcal

\numberwithin{equation}{section}

\theoremstyle{plain}
\newtheorem{theorem}{Theorem}[section]
\newtheorem{proposition}[theorem]{Proposition}

\newtheorem{lemma}[theorem]{Lemma}

\theoremstyle{definition}
\newtheorem{definition}[theorem]{Definition}
\newtheorem{remark}[theorem]{Remark}

\newtheoremstyle{key}{3pt}{3pt}{}{}{\bfseries}{:}{.5em}{}
\theoremstyle{key}
\newtheorem*{JEL}{JEL subject classification}
\newtheorem*{MSC}{MSC2010 subject classification}
\newtheorem*{keywords}{Keywords}

\DeclareMathOperator*{\esssup}{ess\,sup}

\renewcommand{\mid}{\,\vert\,}
\newcommand{\bigv}{\!\bigm\vert\!}

\newcommand{\biggv}{\!\biggm\vert\!}
\providecommand{\abs}[1]{\left\lvert#1\right\rvert}

\newcommand{\nbd}[1]{$#1$\nobreakdash-\hspace{0pt}}
\newcommand{\indi}[1]{\mathbf{1}_{#1}}

\newcommand{\N}{\mathbb{N}}

\newcommand{\R}{\mathbb{R}}

\newcommand{\filt}[1]{\mathbf{#1}} % filtration

\newcommand{\B}{{\mathscr B}}

\newcommand{\F}{{\mathscr F}}

\newcommand{\T}{{\mathscr T}}

\newcommand{\cA}{{\cal A}}
\newcommand{\cC}{{\cal C}}

\newcommand{\cP}{{\cal P}}

\usepackage[breaklinks]{hyperref}

\title{Quick or Persistent? Strategic Investment Demanding Versatility}

\author{
Jan-Henrik Steg\thanks{%
Center for Mathematical Economics, Bielefeld University, Postfach 10 01 31, 33501 Bielefeld, Germany. \texttt{jsteg@uni-bielefeld.de}} 
\and 
Jacco Thijssen\thanks{%
The York Management School, University of York, Freboys Lane, Heslington, York YO10 5GD, UK. \texttt{jacco.thijssen@york.ac.uk}}
}
\date{ }

\begin{document}
\maketitle

\begin{abstract}
In this paper we analyse a dynamic model of investment under uncertainty in a duopoly, in which each firm has an option to switch from the present market to a new market. We construct a subgame perfect equilibrium in mixed strategies and show that both preemption and attrition can occur along typical equilibrium paths. In order to determine the attrition region a two-dimensional constrained optimal stopping problem needs to be solved, for which we characterize the non-trivial stopping boundary in the state space. We explicitly determine Markovian equilibrium stopping rates in the attrition region and show that there is always a positive probability of eventual preemption, contrasting the deterministic version of the model. A simulation-based numerical example illustrates the model and shows the relative likelihoods of investment taking place in attrition and preemption regions.

\begin{keywords}
Stochastic timing games, preemption, war of attrition, real options, {M}arkov perfect equilibrium, two-dimensional optimal stopping.
\end{keywords}
\begin{JEL}
C61, C73, D21, D43, L13
\end{JEL}
\begin{MSC}
60G40, 91A25, 91A55, 91A60
\end{MSC}
\end{abstract}
%\newpage

\section{Introduction}\label{sec:intro}

Consider the following innocuous sounding investment problem involving two coffee shops operating a similar franchise in a town. The profitability of each firm is influenced by some stochastic process, representing uncertainty over the evolution of demand, costs, etc. Each firm has an option to engage in product differentiation by switching to a rival franchise. When should a firm switch franchise, if at all? This economic situation is a very basic example of a \emph{timing game}.\footnote{%
This is a game because the actions of one firm have an influence on the other firm, in that if one coffee shop switches the other firm becomes a (local) monopolist in the original franchise.
} 

In the literature on timing games a distinction is typically made between wars of attrition and preemption, which result from situations with a second-mover or first-mover advantage, respectively. In a game with a first-mover advantage there is a pressure on players to act sooner than is optimal because of the fear of being preempted by the other player. In the classical analysis by \citet{FT85} it is shown that the first-mover advantage of becoming the \emph{leader} is dissipated by preemption, inducing \emph{rent equalization} in that both firms obtain the second-mover, or \emph{follower} payoff in expectation. In games with a second-mover advantage the situation is radically different, because in such games each player wants the other player to move first. This can lead players to take on some costs of waiting in equilibrium, as shown e.g.\ by \citet{HWW88}.

Timing games have been used for instance in the industrial organization literature to explain a plethora of phenomena dealing with such issues as technology adoption \citep{FT85}, firms' exit decisions in a duopoly \citep{FT86}, and patent races (\citealp{Rein82}). Relatively early on in the development of the real options theory of investment under uncertainty it was recognized that real options, unlike their financial counterparts, are usually non-exclusive and that, thus, a game theoretic approach is needed. Many contributions in this literature have focussed on analysing markets with a first-mover advantage.\footnote{%
A non-exhaustive list of contributions is \citet{Smets91}, \citet{HK99}, \citet{THK12}, \citet{Weeds02}, \citet{BLMM04}, \citet{PK06}. Extensive surveys can be found in \citet{CT11} and \citet{AP14}.
}
The literature on game theoretic real options models with second-mover advantages is less well-developed, with \citet{Hoppe00} and \citet{Mur04} being notable exceptions.

In this paper we want to think about duopolies in which \emph{both} attrition and preemption can occur, depending on the evolution of the market environment, and to construct a subgame perfect equilibrium for such games. It is in such a model that the true power of stochastic versus deterministic models becomes clear. Typically, in models with either a first- or second-mover advantage, regardless of whether the model is deterministic or stochastic, it is clear \emph{a priori} what will happen in equilibrium: preemption in case of a first-mover and attrition in the case of a second-mover advantage. It is just the timing that is different: on average investment takes place later in a stochastic model. However, in many real-world situations it not clear whether a market will develop into one with a first- or second-mover advantage. In such cases one has to resort to a stochastic model.

Returning to the coffee franchise example, a switch by either firm changes the firms' profitabilities in two ways: (i) there is a change in market structure and (ii) the stochastic shocks under the new franchise follow a different process. The economic conflict here is the influence that each firm has over its competitor by switching to the other franchise. If only one firm switches, then each firm will operate as a monopolist in its respective market. This is good for both firms, but potentially best for the firm that does not switch, because it does not incur the sunk costs. This could lead to a war of attrition, i.e.\ a situation where it is profitable for each firm to switch, but each prefers the other firm to do so. On the other hand, if the profitability in the new market is higher than in the current market but not high enough to make a duopoly in the new market profitable, then a situation of preemption may occur. In such cases firms may try to switch before it is optimal to do so.

The novelty of the model we develop is that it can deal with attrition and preemption simultaneously. We do this by using a two-dimensional state space, roughly speaking corresponding to the profitability in each of the two markets. In equilibrium this state space will effectively be split into three regions: a \emph{preemption region} where both players try to preempt each other, an \emph{attrition region} where each firm prefers its competitor to invest, and a \emph{continuation region} where neither firm acts. Along a sample path the following may happen. If the state hits the preemption region, then at least one firm invests instantly. If the state moves through the attrition region, each firm invests at a certain rate (the \emph{attrition rate}) until either the preemption region is hit, or the attrition region is left towards the continuation region. Inside the latter nothing happens until either the preemption region is hit directly or the attrition region (again).

The richness of this model does present some technical challenges. Firstly, in regimes with a second-mover advantage, payoffs of pure-strategy equilibria are generally asymmetric and each firm prefers to become follower. Determining the roles is then a (commitment) problem left open. By considering mixed strategies we can make the firms indifferent. With uncertainty, however, the attrition region where the mixing occurs can be entered and left repeatedly at random. We nevertheless obtain an equilibrium with Markovian stopping probabilities per small time interval (the so-called \emph{attrition rate}), while deterministic models typically contain some simplifying smoothness assumption on the payoff functions. In the latter case the evolution of the game is not only perfectly predictable, but also ``slow''. Thus, when the attrition rate becomes unbounded near the preemption region as the second-mover advantage vanishes, stopping occurs a.s.\ before reaching the preemption region. With our faster stochastic dynamics, in contrast, there is always a positive probability of preemption to actually occur.

Another important issue relates to the value functions of the leader and follower roles. Along a sample path, the game ends for sure as soon as the preemption region is entered. This result has been known since the \citet{FT85} analysis of a deterministic preemption game. In our paper subgame perfection leads to similar equilibrium behavior. Consequently, firms know \emph{ex ante} that the game ``ends'' as soon as the preemption region is first hit. Working backwards, this means that firms must take this into account when they choose their strategy in the attrition region. In particular, as we will see, this has an effect on the attrition region itself. Determining this region will require solving a genuinely two-dimensional constrained optimal stopping problem. There are no known general techniques to solve such problems (ours \emph{cannot} be transformed into a one-dimensional problem like some models considered in the optimal stopping literature), but we are able to show that, quite remarkably, our model has enough structure to fully characterize the attrition and preemption regions.

Finally, we propose a simulation-based approach to numerically analyse the model. This simulation shows, for a particular starting point in the continuation region, that many sample paths actually reach the preemption region via the attrition region. This shows that our results on stochastic versus deterministic models are economically important.

The paper is organized as follows. In Section~\ref{sec:model} we present a stochastic dynamic model of a market in which both first- and second-mover advantages can occur. In Section~\ref{sec:informal} we describe the paper's results in an informal way. The timing game is formally defined, analysed and discussed in detail in Section~\ref{sec:formal}. In particular we construct a subgame perfect equilibrium in mixed (Markovian) strategies in Section~\ref{subsec:eql}. An important part of this construction consists of solving a constrained optimal stopping problem in Section~\ref{subsec:stop}. Section~\ref{sec:example} present a numerical example that we use to discuss the economic content of our results and in Section~\ref{sec:dtm} we contrast our stochastic analysis with its deterministic version. Finally, Section~\ref{sec:concl} provides some concluding remarks.

\section{A Model with First- and Second-Mover Advantages: Product Differentiation in a Duopoly}\label{sec:model}

Consider two firms that are producing a homogeneous good. For instance, the firms could be two coffee shops operating a similar franchise $A$ in a town. The profitability of each firm is influenced by some stochastic process. Each firm has an option to engage in product differentiation and switch to a different franchise $B$. This switch involves some sunk costs (for example due to refurbishment). To make matters more concrete, suppose that when firms operate the same franchise Bertrand competition implies that both firms make no profits. If, however, firms operate a different franchise both can be seen as monopolists in their respective markets and they will earn some profit streams. The monopoly profits in franchises $A$ and $B$ are given by the processes $X=(X_t)_{t\geq 0}$ and $Y=(Y_t)_{t\geq 0}$, respectively, while there are no profits in duopoly in either franchise. There is a sunk cost of switching $I>0$. For completeness we consider some additional running costs before and after switching denoted by $c_0$, $c_A$ and $c_B$, respectively, but these have no qualitative impact and may also be ignored. Depending on the net present value of the two monopoly profits $X$ and $Y$ and the total costs of switching, it may at any time be more profitable to switch to the new market, to become monopolist in the current market, or to stay in duopoly.

To obtain some explicit results we assume that $X$ and $Y$ are geometric Brownian motions, specifically as the unique solutions to the system of stochastic differential equations
\begin{equation}\label{SDE}
\frac{dX}{X}=\mu_Xdt+\sigma_XdB^X\quad\text{and}\quad\frac{dY}{Y}=\mu_Ydt+\sigma_YdB^Y
\end{equation}
with given initial values $(X_0,Y_0)=(x,y)\in\R_+^2$. $B^X$ and $B^Y$ are Brownian motions that may be correlated with coefficient $\rho$, and $\mu_X$, $\mu_Y$, $\sigma_X$ and $\sigma_Y$ are some constants representing the drift and volatility of the growth rates of our profit processes $X$ and $Y$, respectively. We assume $\sigma_X^2,\sigma_Y^2>0$ and $\abs{\rho}<1$ to exclude potential degeneracies that require separate treatment (the case $\sigma_X=\sigma_Y=0$ is considered in Section \ref{sec:dtm}). We also assume that profits are discounted at a common and constant rate $r>\max(0,\mu_X,\mu_Y)$ to ensure finite values of the following payoff processes and the subsequent stopping problems.\footnote{\label{fn:XclassD}%
Then $(e^{-rt}X_t)$ is bounded by an integrable random variable. Indeed, for $\sigma_X>0$ we have $\sup_te^{-rt}X_t=X_0e^{\sigma_XZ}$ with $Z=\sup_tB^X_t-t(r-\mu_X+\sigma_X^2/2)/\sigma_X$, which is exponentially distributed with rate $2(r-\mu_X)/\sigma_X+\sigma_X$ (see, e.g., \cite{RevuzYor}, Exercise (3.12) $4^\circ$). Thus, $E[\sup_te^{-rt}X_t]=X_0(1+\sigma_X^2/2(r-\mu_X))\in\R_+$, implying that $(e^{-rt}X_t)$ is of class (D); analogously for $\sigma_X<0$ and $Y$.
}

Following the previous discussion, the basic payoffs can now be formulated by processes $L$, $F$ and $M$, which represent the expected payoffs at the time of the (first and only) switch to \emph{(i)} a firm that switches solely and becomes the \emph{leader}, \emph{(ii)} a firm that remains and becomes the \emph{follower} and \emph{(iii)} a firm that switches simultaneously with the other, respectively.

At any time $t\geq 0$ these processes take the values
\begin{flalign}\label{L(y)}
&& L_t:={}&-\int_0^te^{-rs}c_0\,ds+E\biggl[\int_t^\infty e^{-rs}(Y_s-c_B)\,ds\biggv\F_t\biggr]-e^{-rt}I &&\nonumber\\
&& ={}&-\frac{c_0}{r}+e^{-rt}\biggl(\frac{Y_t}{r-\mu_Y}-\frac{c_B-c_0}{r}-I\biggr):=L(t,Y_t), &&\nonumber\\
&& F_t:={}&-\int_0^te^{-rs}c_0\,ds+E\biggl[\int_t^\infty e^{-rs}(X_s-c_A)\,ds\biggv\F_t\biggr] &&\nonumber\\
&& ={}&-\frac{c_0}{r}+e^{-rt}\biggl(\frac{X_t}{r-\mu_X}-\frac{c_A-c_0}{r}\biggr):=F(t,X_t) &&
\end{flalign}
and
\begin{flalign*}
&& M_t:={}&-\int_0^\infty e^{-rs}c_0\,ds-e^{-rt}I\hphantom{+E\biggl[\int_t^\infty e^{-rs}(Y_s-c_B)\,ds\biggv\F_t\biggr]} &&\\
&& ={}&-\frac{c_0}{r}-e^{-rt}I. &&
\end{flalign*}
Note that these processes are continuous thanks to their respective representation by a continuous function of the underlying continuous profit processes $X$ and $Y$. By those  Markovian representations we may also evaluate the payoff processes at any stopping time at which the switch might occur to obtain the correct payoffs.\footnote{%\label{fn:Lcont}%
If the switch occurs at a stopping time $\tau$, one has to take conditional expectations at $\tau$ to determine the appropriate payoff $L_\tau$ for example, which in general need not be consistent at all with a family of conditional expectations indexed by deterministic times $t$ and a pointwise construction via $t=\tau(\omega)$.
}
All payoff processes are further bounded by integrable random variables, cf.\ footnote \ref{fn:XclassD}, and they converge to $-c_0/r=:L_\infty$ as $t\to\infty$.

We assume that it is optimal for only one firm to switch, specifically $(c_0-c_A)/r+I\geq 0$ to ensure $F\geq M$, and that the relative capitalized total cost of a single firm that switches is nonnegative, i.e., $(c_B-c_A)/r+I\geq 0$.

The firms will try to preempt each other in switching to the franchise $B$ when $L>F$, which is the case iff
\begin{equation*}
(X,Y)\in\cP:=\Bigl\{(x,y)\in\R_+^2\Bigm\vert y>(r-\mu_Y)\Bigl(\frac{x}{r-\mu_X}+\frac{c_B-c_A}{r}+I\Bigr)\Bigr\}.
\end{equation*}
We accordingly call $\cP$ the \emph{preemption region} of the state space $\R_+^2$ of our process $(X,Y)$. When $F>L$, each firm prefers to be the one who stays if there is a switch, which potentially induces a war of attrition. We will model this strategic conflict as a timing game.

\section{An Informal Preview of the Results}\label{sec:informal}

At this stage it is instructive to analyse the model informally and anticipate the equilibrium construction that follows. To begin, in line with the literature on preemption games our game has to end as soon as the set $\cP$ is hit. However, to model this preemption outcome it is not enough to consider distribution functions over time as (mixed) strategies.\footnote{%
See \cite{HendricksWilson92} for a proof of equilibrium nonexistence.
}
As there is no ``next period'' in continuous time, one has to enable the firms both to try to invest immediately but avoiding simultaneous investment at least partially, in particular on the boundary of $\cP$, where they are still indifferent. We follow the endogenous approach of \cite{FT85}, augmenting strategies by ``intensities'' $\alpha_i\in[0,1]$ that determine the outcome when both firms try to invest simultaneously like in an infinitely repeated grab-the-dollar game. A firm that grabs first invests. If firm $j$ grabs with stationary probability $\alpha_j>0$, firm $i$ can obtain the follower payoff $F$ by never grabbing, and $\alpha_jM+(1-\alpha_j)L$ by grabbing with probability 1. Hence the firms are just indifferent between both actions in equilibrium if
\begin{equation*}
\alpha_1=\alpha_2=\alpha=\frac{L-F}{L-M}\in(0,1],
\end{equation*}
implying the expected local payoffs $F$ (note that we have $F\geq M$). With symmetric intensities the probability that either firm becomes leader is then
\begin{equation*}
\alpha(1-\alpha)+(1-\alpha)^2\alpha(1-\alpha)+\cdots=\frac{1-\alpha}{2-\alpha}.
\end{equation*}
Taking limits, that probability becomes $\frac12$ on the boundary of $\cP$ where $\alpha$ vanishes (see Remark \ref{rem:outcome} in the Appendix on the limit outcome). In the interior of $\cP$, however, there is a positive probability $\alpha/(2-\alpha)$ of simultaneous investment representing the cost of preemption, in contrast to more ad hoc coordination devices.

Now consider states outside $\cP$. Suppose first that some firm tries to determine when it would be optimal to make the switch and obtain the leader value $L$, which directly depends only on $Y$. If the competitor could not preempt the firm, standard techniques would yield that there is a threshold $y^{\ast}$ that separates the \emph{continuation region}, where the firm remains inactive, and the \emph{stopping region} where the firm makes the switch. The optimal stopping time is thus the first time the process $Y$ enters the set $[y^{\ast},\infty)$ and the process $X$ plays no role in this problem at all.

The firm should realize, however, that as soon as $\cP$ is hit the game is over and it receives, in expectation, the follower value, which does depend on $X$. Thus the firm should actually solve the \emph{constrained} optimal stopping problem to become leader up to the first hitting time of $\cP$. Intuitively, the solution to this problem should again divide the (genuinely two-dimensional) state space into two sets: a continuation set $\cC$ and a stopping set $\cA$. The boundary between these regions turns out to be given by a mapping $x\mapsto b(x)$, see Figure~\ref{fig:regions}. $b(x)$ increases to $y^\ast$ for $x\to\infty$, since then the probability of reaching $\cP$ while waiting to get closer to the unconstrained threshold $y^\ast$ is small.

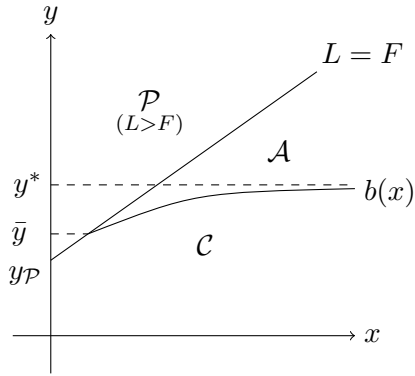
\begin{figure}[ht]
  \centering
  \begin{tikzpicture}[inner sep=0pt,minimum size=0pt,label distance=3pt]
    \draw[->] (-0.5,0) -- (4,0) node[label=right:$x$] {}; % x axis
    \draw[->] (0,-0.5) -- (0,4) node[label=above:$y$] {}; % y axis
    \draw[dashed] (0,2) -- (4,2) [] {};
    \draw[-] (0,1) -- (3.5,3.5) []{};
    \draw[-] (0.49,1.35) .. controls (1.9,1.9) and (1.9,1.9) .. (4,1.95) [] {};
    \draw[dashed] (0,1.35) -- (.49,1.35) [] {};

    %\fill[black] (1.5,1.85) circle (.04);

    \node at (0,2) [label=left:$y^*$] {};
    \node at (0,1.35) [label=left:${\bar y\ \,}$] {};
    \node at (0,.8) [label=left:${y_\cP}$] {};
    \node at (4.1,3.5) [label=above:${L=F}$] {};
    \node at (1.3,2.5) [label=above:$\underset{(L>F)}{\cP}$] {};
    \node at (3,2.2) [label=above:$\cA$] {};
    \node at (4,1.9) [label=right:$b(x)$] {};
    \node at (1.8,1.2) [label=right:$\cC$] {};
  \end{tikzpicture}
  \caption{Continuation, Attrition and Preemption regions.\protect\footnotemark}
  \label{fig:regions}
\end{figure}
\footnotetext{%
The definitions of $\bar{y}$ and $y_\cP$ will become clear in Section~\ref{subsec:stop}.
}%

For smaller values of $x$, however, there is a higher risk to get trapped in preemption if $\cP$ is hit with $Y<y^\ast$. Thus it may be better to secure the current leader value before $\cP$ is hit at a possibly even lower value. We formally show that it is optimal to stop strictly before $\cP\cup\R_+\times[y^\ast,\infty)$ is hit in the constrained stopping problem, i.e., the stopping boundary $b(x)$ lies below that region (in particular when preemption is ``near''). In the interior of the stopping set $\cA$ we now have a situation alike a war of attrition: waiting without a switch occuring is costly (the constrained leader value decreases in expectation), but becoming follower in this region would yield a higher payoff $F>L$. 

We characterize the attrition region $\cA$ below in terms of the stopping boundary $b(x)$ for the constrained leader's stopping problem and identify the cost of waiting in this region as the \emph{drift} of $L$. In equilibrium the firms are only willing to continue and bear the cost in $\cA$ if there is a chance that the opponent drops out, yielding the follower's payoff. We will propose symmetric Markovian equilibrium hazard rates of investing~-- called \emph{attrition rates}~-- that make both firms exactly indifferent to continue in $\cA$. If we denote those rates here by $\lambda_t$ to give a heuristic argument, the probability that a firm invests in a small time interval $[t,t+dt]$ is $\lambda_t\,dt$. Since each firm is supposed to be indifferent to wait, the equilibrium value (i.e., expected payoff) $V_t$ should satisfy
\begin{align*}
V_t&=F_t\lambda_t\,dt+(1-\lambda_t\,dt)E[V_{t+dt}\mid\F_t]=F_t\lambda_t\,dt+(1-\lambda_t\,dt)E[V_t+dV_t\mid\F_t] \\
\Leftrightarrow\quad\lambda_t\,dt&=\frac{-E[dV_t\mid\F_t]}{F_t-V_t-E[dV_t\mid\F_t]}.
\end{align*}
On the other hand we should have $V$=$L$ in $\cA$ by indifference and since the probability of simultaneous investment is zero if the other firm invests at a rate. Hence $E[dV_t\mid\F_t]$ is the (here negative) drift of $L_t$ of order $o(dt)$, implying
\begin{align*}
\lambda_t\,dt=\frac{-E[dL_t\mid\F_t]}{F_t-L_t}.
\end{align*}

We will model strategies more generally as distribution functions $G_i(t)$ over time, such that the hazard rates are actually $\lambda_t\,dt=dG_i(t)/(1-G_i(t))$. Note that in contrast to deterministic models that develop linearly (with time $t$ the state variable), we have to account for the possibility that the state enters and leaves the attrition region frequently and randomly, which makes indifference much more complex to verify. A distinctive feature of the exogenous uncertainty is that there is always a positive probability of reaching the preemption region with no firm having invested before, although the hazard rates $\lambda_t$ grow unboundedly as $F-L\to 0$ near $\cP$; the latter does enforce investment in the deterministic analogue.

\section{Formal Analysis of Subgame Perfect Equilibria}\label{sec:formal}

In this section we present and discuss in more detail the formal results that have been introduced informally in Section \ref{sec:informal}. We start by formalizing the timing game (in particular strategies and equilibrium concept) in Section \ref{subsec:game}, followed by establishing equilibria for subgames starting in the preemption region $\cP$ in Section \ref{subsec:preemeql}. The constrained optimal stopping problem that identifies the attrition region $\cA$ is analysed in Section \ref{subsec:stop}. In Section \ref{subsec:eql} we establish (Markovian) equilibrium strategies for arbitrary subgames.

\subsection{Timing Games: Strategies and Equilibrium}\label{subsec:game}

We use the framework of subgame perfect equilibrium for stochastic timing games with mixed strategies developed in \cite{RiedelSteg14}. There is a fixed probability space $\left(\Omega, \F, P\right)$ capturing uncertainty about the state of the world, in particular concerning future profits. Our firms can decide to switch to the new market $B$ in continuous time $t\in\R_+$. The decision when to switch can be based on some exogenous dynamic information represented by a filtration $\filt{F}=\bigl(\F_t\bigr)_{t\geq 0}$ satisfying the usual conditions (i.e., $\filt{F}$ is right-continuous and complete). It includes the current observations of the profit processes $X$ and $Y$, hence $L$, $F$ and $M$ are adapted to $\filt{F}$, too.

The feasible decision nodes in continuous time are all stopping times. Therefore any stopping time $\vartheta$ represents the beginning of a subgame, with the connotation that no firm has switched before. Let $\T$ denote the set of all stopping times w.r.t.\ the filtration $\filt{F}$.

We will specify complete plans of actions for all subgames, taking the form of (random) distribution functions over time. These have to be time-consistent, meaning that Bayes' rule has to hold wherever applicable. Additional strategy extensions serving as a coordination device are needed to model preemption appropriately in continuous time.

\begin{definition}\label{def:alpha}
An \emph{extended mixed strategy} for firm $i\in\{1,2\}$ in the subgame starting at $\vartheta\in\T$, also called \emph{\nbd{\vartheta}strategy}, is a pair of processes $\bigl(G^\vartheta_i,\alpha^\vartheta_i\bigr)$ taking values in $[0,1]$, respectively, with the following properties.
\begin{enumerate}
\item
$G^\vartheta_i$ is adapted. It is right-continuous and nondecreasing with $G^\vartheta_i(t)=0$ for all $t<\vartheta$, a.s.

\item
$\alpha^\vartheta_i$ is progressively measurable.\footnote{%
Formally, the mapping $\alpha^\vartheta_i\colon\Omega\times[0,t]\to\R$, $(\omega,s)\mapsto\alpha^\vartheta_i(\omega,s)$ must be $\F_t\otimes\B([0,t])$-measurable for any $t\in\R_+$. It is a stronger condition than adaptedness, but weaker than optionality, which we automatically have for $G^\vartheta_i$ by right-continuity. Progressive measurability implies that $\alpha^\vartheta_i(\tau)$ will be \nbd{\F_\tau}measurable for any $\tau\in\T$.
}
It is right-continuous where $\alpha^\vartheta_i<1$, a.s.\footnote{%
This means that with probability 1, $\alpha^\vartheta_i(\cdot)$ is right-continuous at all $t\in[0,\infty)$ for which $\alpha^\vartheta_i(t)<1$. Since we are here only interested in \emph{symmetric} games, we may demand the extensions $\alpha^\vartheta_i(\cdot)$ to be right-continuous also where they take the value 0, which simplifies the definition of outcomes. See Section 3 of \cite{RiedelSteg14} for issues with asymmetric games and corresponding weaker regularity restrictions.
}

\item
\begin{equation*}
\alpha^\vartheta_i(t)>0\Rightarrow G_i^\vartheta(t)=1\qquad\text{for all }t\geq 0\text{, a.s.}
\end{equation*}

\end{enumerate}
We further define $G^\vartheta_i(0-)\equiv 0$, $G^\vartheta_i(\infty)\equiv 1$ and $\alpha^\vartheta_i(\infty)\equiv 1$ for every extended mixed strategy.
\end{definition}

Any stopping time $\tau\geq\vartheta$ can be interpreted as a pure strategy corresponding to $G^\vartheta_i(t)=\indi{t\geq\tau}$ for all $t\geq 0$. If the $G^\vartheta_i$ jump to 1 simultaneously, the extensions $\alpha^\vartheta_i$ determine the  actual outcome probabilities, so that simultaneous stopping is avoidable to a certain extent. The $\alpha$-components were first introduced by \citet{FT85} for deterministic games and allow for instantaneous coordination in continuous time. They can also be thought of as measuring the ``preemption intensity'' of firms as they determine the probabilities with which each firm succeeds in investing. In particular, at $\hat\tau^\vartheta:=\inf\{t\geq\vartheta\mid\alpha^\vartheta_1(t)+\alpha^\vartheta_2(t)>0\}$, the extensions $\alpha^\vartheta_\cdot$ determine final outcome probabilities $\lambda^\vartheta_{L,i}$, $\lambda^\vartheta_{L,j}$ and $\lambda^\vartheta_{M}$ as defined in \cite{RiedelSteg14}, which we provide for completeness in Appendix \ref{app:outcome}. Here $\lambda^\vartheta_{L,i}$ and $\lambda^\vartheta_{L,j}$ denote the probabilities that firm~$i$ becomes the leader or follower, respectively. $\lambda^\vartheta_{M}$ is the probability that both firms invest simultaneously. Since the latter outcome is never desirable (in the sense that firms would rather be the follower than end up in a situation of simultaneous investment) this outcome has traditionally been referred to as a ``coordination failure'' (see, for example, \citealp{HK99}).

Now that the strategies have been formalized, we can define the firms' payoffs. As with strategies, these have to be defined for every subgame.

\begin{definition}\label{def:payoffs_extended}
Given two extended mixed strategies $\bigl(G^\vartheta_i,\alpha^\vartheta_i\bigr)$, $\bigl(G^\vartheta_j,\alpha^\vartheta_j\bigr)$, $i,j\in\{1,2\}$, $i\not=j$, the \emph{payoff} of firm $i$ in the subgame starting at $\vartheta\in\T$ is
\begin{align*}%\label{Vi_mixed}
V^\vartheta_i\bigl(G^\vartheta_i,\alpha^\vartheta_i,G^\vartheta_j,\alpha^\vartheta_j\bigr):=E&\biggl[\int_{[0,\hat\tau^\vartheta)}\bigl(1-G^\vartheta_j(s)\bigr)L_s\,dG^\vartheta_i(s)\nonumber\\
&+\int_{[0,\hat\tau^\vartheta)}\bigl(1-G^\vartheta_i(s)\bigr)F_s\,dG^\vartheta_j(s)\nonumber\\
&+\sum_{s\in[0,\hat\tau^\vartheta)}\Delta G^\vartheta_i(s)\Delta G^\vartheta_j(s)M_s\nonumber\\
&+\lambda^\vartheta_{L,i}L_{\hat\tau^\vartheta}+\lambda^\vartheta_{L,j}F_{\hat\tau^\vartheta}+\lambda^\vartheta_{M}M_{\hat\tau^\vartheta}\biggv\F_\vartheta\biggr].
\end{align*}
\end{definition}

Since \nbd{\vartheta}strategies need to be aggregated across subgames, we need to introduce a notion of time consistency.
\begin{definition}\label{def:TC_extended}
An \emph{extended mixed strategy} for firm $i\in\{1,2\}$ in the timing game is a family
\begin{equation*}
\bigl(G_i,\alpha_i\bigr):=\bigl(G_i^\vartheta,\alpha^\vartheta_i\bigr)_{\vartheta\in\T}
\end{equation*}
of extended mixed strategies for all subgames $\vartheta\in\T$.

An extended mixed strategy $\bigl(G_i,\alpha_i\bigr)$ is \emph{time-consistent} if for all $\vartheta\leq\vartheta'\in\T$
\begin{flalign*}
&& \vartheta'\leq t\in\R_+ &\Rightarrow\ G_i^\vartheta(t)=G_i^\vartheta(\vartheta'-)+\bigl(1-G_i^\vartheta(\vartheta'-)\bigr)G_i^{\vartheta'}(t)\quad\text{a.s.} && \\
&\text{and} \\
&& \vartheta'\leq\tau\in\T &\Rightarrow\ \alpha^\vartheta_i(\tau)=\alpha^{\vartheta'}_i(\tau)\quad\text{a.s.} &&
\end{flalign*}
\end{definition}

Now the equilibrium concept is standard.

\begin{definition}\label{def:SPE_extended}
A \emph{subgame perfect equilibrium} for the timing game is a pair $\bigl(G_1,\alpha_1\bigr)$, $\bigl(G_2,\alpha_2\bigr)$ of time-consistent extended mixed strategies such that for all $\vartheta\in\T$, $i,j\in\{1,2\}$, $i\not=j$, and extended mixed strategies $\bigl(G_a^\vartheta,\alpha^\vartheta_a\bigr)$
\begin{equation*}
V_i^\vartheta(G_i^\vartheta,\alpha^\vartheta_i,G_j^\vartheta,\alpha^\vartheta_j)\geq V_i^\vartheta(G_a^\vartheta,\alpha^\vartheta_a,G_j^\vartheta,\alpha^\vartheta_j)\quad\text{a.s.},
\end{equation*}
i.e., such that every pair $\bigl(G^\vartheta_1,\alpha^\vartheta_1\bigr)$, $\bigl(G^\vartheta_2,\alpha^\vartheta_2\bigr)$ is an \emph{equilibrium} in the subgame at $\vartheta\in\T$, respectively.
\end{definition}

\subsection{Preemption Equilibria}\label{subsec:preemeql}

We start our construction of a subgame perfect equilibrium by considering subgames with a first-mover advantage $L_\vartheta>F_\vartheta$, i.e., starting in the preemption region $\cP$, where we can refer to the following established equilibria in which at least one firm switches immediately.

\begin{proposition}[{\cite{RiedelSteg14}, Proposition 3.1}]\label{prop:eqlL>F}
Fix $\vartheta\in\T$ and suppose $\vartheta=\inf\{t\geq\vartheta\mid L_t>F_t\}$ a.s. Then $\bigl(G^\vartheta_1,\alpha^\vartheta_1\bigr)$, $\bigl(G^\vartheta_2,\alpha^\vartheta_2\bigr)$ defined by
\begin{equation*}
\alpha^\vartheta_i(t)=\indi{L_t>F_t}\frac{L_t-F_t}{L_t-M_t}
\end{equation*}
for any $t\in[\vartheta,\infty)$ and $G^\vartheta_i=\indi{t\geq\vartheta}$, $i=1,2$, are an equilibrium in the subgame at $\vartheta$.

The resulting payoffs are $V_i^\vartheta(G_i^\vartheta,\alpha^\vartheta_i,G_j^\vartheta,\alpha^\vartheta_j)=F_\vartheta$.
\end{proposition}

In these equilibria, the firms are indifferent between stopping and waiting. The latter would mean becoming follower instantaneously, as then the opponent switches with certainty. If $L_\vartheta>F_\vartheta$, then there is a positive probability of simultaneous switching, which is the ``cost of preemption'', driving the payoffs down to $F_\vartheta$. Of particular interest are however subgames $\vartheta$ with $L_\vartheta=F_\vartheta$, where each firm becomes leader or follower with probability $\frac12$ in the equilibrium of Proposition~\ref{prop:eqlL>F} (given the outcome probabilities defined in Section \ref{app:outcome}). This is in particular the case in ``continuation'' equilibria at $\tau_\cP(\vartheta):=\inf\{t\geq\vartheta\mid L_t>F_t\}$ (where the game will hence end with probability 1) for subgames which begin with $L_\vartheta\leq F_\vartheta$. By these continuation equilibria the  firms know that \emph{if} the preemption region is ever reached, then each firm can expect to earn the follower value \emph{at that time}. This observation will turn out to be important for constructing the equilibrium in the case of attrition.

\subsection{Constrained Leader's Stopping Problem}\label{subsec:stop}

We now turn to equilibria for subgames beginning outside the preemption region where we will observe a war of attrition in contrast to typical strategic real option models in the literature, in which only preemption is considered. Our equilibria are closely related to a particular stopping problem that we hence discuss in detail. Its value (function) will also be the player's continuation value in equilibrium and its solution is needed to characterize and understand equilibrium strategies. This stopping problem is also of interest on its own as it is two-dimensional and therefore not at all standard in the optimal stopping literature.

In accordance with Proposition \ref{prop:eqlL>F} we fix the (equilibrium) payoff $F$ whenever the state hits the preemption region $\cP$. Suppose now that only one firm, say $i$, can make the switch before $\cP$ is reached. Then $i$ can determine when to become optimally the leader up to hitting $\cP$, where the game ends with the current value of $F$ as the expected payoff.

Letting $\tau_\cP:=\inf\{t\geq 0\mid(X_t,Y_t)\in\cP\}=\inf\{t\geq 0\mid L_t>F_t\}$ denote the first hitting time of the preemption region $\cP$, firm $i$ now faces the problem of optimally stopping the auxiliary payoff process
\begin{equation*}
\tilde L:=L\indi{t<\tau_\cP}+F_{\tau_\cP}\indi{t\geq\tau_\cP}.
\end{equation*}
It obviously suffices to consider stopping times $\tau\leq\tau_\cP$. We are also only interested in initial states $(X_0,Y_0)=(x,y)\in\cP^c$, which implies $F_{\tau_\cP}=L_{\tau_\cP}$ by continuity. Then the value of our stopping problem is
\begin{equation*}%\label{eq:constrained}
V_{\tilde L}(x,y):=\esssup_{\tau\geq 0}E\bigl[\tilde L_\tau\bigr]=\esssup_{\tau\in[0,\tau_\cP]}E\bigl[L_\tau\bigr]
\end{equation*}
for $(X_0,Y_0)=(x,y)\in\cP^c$ (and $V_{\tilde L}(x,y):=F_0$ for $(X_0,Y_0)=(x,y)\in\cP$). Thanks to the strong Markov property, the solution of this problem can be characterized by identifying the \emph{stopping region} of the state space $\{(x,y)\in\R_+^2\mid V_{\tilde L}(x,y)=\tilde L_0\text{ for }(X_0,Y_0)=(x,y)\}=\cP\cup\{(x,y)\in\cP^c\mid V_{\tilde L}(x,y)=L(0,y)\}$ and the \emph{continuation region} 
\begin{equation*}
\cC:=\{(x,y)\in\cP^c\mid V_{\tilde L}(x,y)>L(0,y)\}\subset\cP^c.
\end{equation*} 
By the continuity of $\tilde L$ (resp.\ $L$) it is indeed optimal to stop as soon as $(X,Y)$ hits $\cC^c$.

The stopping and continuation regions are related to those for the unconstrained problem $\sup_{\tau\geq 0}E\bigl[L_\tau\bigr]$, which only depends on $Y$ and its initial value $Y_0=y$. This is a standard problem of the real options literature, e.g., which is uniquely solved\footnote{%
The solution also holds in the degenerate case $Y_0=y^*=0$, iff $L$ is constant, but then it is not unique, of course.
}
by stopping the first time $Y$ exceeds the threshold
\begin{equation*}
y^*=\frac{\beta_1}{\beta_1-1}(r-\mu_Y)\Bigl(I+\frac{c_B-c_0}{r}\Bigr),
\end{equation*}
with $\beta_1>1$ the positive root of the quadratic equation
\begin{equation}\label{eq:quadratic}
 \mathcal{Q}(\beta)\equiv\frac{1}{2}\sigma_Y^2\beta(\beta-1)+\mu_Y\beta-r=0.
\end{equation}
At $y^*$, the net present value of investing in $B$ as a monopolist is just high enough to offset the option value of waiting.\footnote{%
Above $y^*$, the forgone revenue from any delay is so high that investment is strictly optimal. Below $y^*$, the depreciation effect on the investment cost is dominant to make waiting strictly optimal.
}

Now the constraint is binding iff the stopping region for the unconstrained problem does \emph{not} completely contain the preemption region $\cP$, which is the situation depicted in Figure~\ref{fig:regions} above with
\begin{equation*}
y^*>y_\cP:=(r-\mu_Y)\biggl(\frac{c_B-c_A}{r}+I\biggr).\protect\footnotemark
\end{equation*}
\footnotetext{%
$y_\cP\geq 0$ by our assumption on the parameters.
}%
Indeed, as stopping $L$ is optimal for $Y_t\geq y^*$ in the unconstrained case, it is too under the constraint $\tau\leq\tau_\cP$. Hence, if $y^*\leq y_\cP$, the continuation regions for the constrained and unconstrained problems agree.

%\begin{figure}[ht]
%  \centering
%  \begin{tikzpicture}[inner sep=0pt,minimum size=0pt,label distance=3pt]
%    \draw[->] (-0.5,0) -- (4,0) node[label=right:$x$] {}; % x axis
%    \draw[->] (0,-0.5) -- (0,4) node[label=above:$y$] {}; % y axis
%    \draw[dashed] (0,2) -- (4,2) [] {};
%%   \draw[-] (4/3,2) -- (4,2) [] {};
%    \draw[-] (0,1) -- (3.5,3.5) []{};
%    \draw[-] (0.49,1.35) .. controls (1.9,1.9) and (1.9,1.9) .. (4,1.95) [] {};
%    \draw[dashed] (0,1.35) -- (.49,1.35) [] {};
%
%    \fill[black] (1.5,1.85) circle (.04);
%
%    \node at (0,2) [label=left:$y^*$] {};
%    \node at (0,1.35) [label=left:${\bar y\ \,}$] {};
%    \node at (0,.8) [label=left:${y_\cP}$] {};
%    \node at (4.1,3.5) [label=above:${L=F}$] {};
%    \node at (1.3,2.5) [label=above:$\underset{(L>F)}{\cP}$] {};
%    \node at (3,2.2) [label=above:$\cA$] {};
%    \node at (4,1.9) [label=right:$b(x)$] {};
%    \node at (1.8,1.2) [label=right:$\cC$] {};
%  \end{tikzpicture}
%  %\end{flushleft}
%  \caption{Attrition and Preemption regions.}
%  \label{fig:regions}
%\end{figure}

In the displayed case $y^*>y_\cP$, however, we have a much more complicated, truly two-dimensional stopping problem.\footnote{%
Genuinely two-dimensional problems and their free boundaries are rarely studied in the literature due to their general complexity. Sometimes problems that appear two-dimensional are considered, e.g.\ involving a one-dimensional diffusion and its running supremum, which are then reduced to a one-dimensional, more standard problem. See, e.g., \cite{PS06}.
}
Then the stopping region for $\tilde L$ extends below $\cP\cup\{(x,y)\in\R_+^2\mid y\geq y^*\}$ due to the risk of hitting $\cP$ at a lower value of $Y$ when $(X_0,Y_0)$ is close to $\cP$ but below $y^*$. Stopping is dominated, however, below the value
\begin{equation*}
\bar y:=c_B-c_0+rI,
\end{equation*}
where the drift of $L$ is positive, whence also $\bar y<y^*$ if either is positive.\footnote{%
The drift of $L$, $-e^{-rt}(Y_t-\bar y)\,dt$, can be derived by applying It\=o's formula to $L(t,Y_t)$ in \eqref{L(y)}. With $\bar y$ we also have
\begin{equation*}
L_t=-\frac{c_0}{r}+E\biggl[\int_t^\infty e^{-rs}(Y_s-\bar y)\,ds\biggv\F_t\biggr].
\end{equation*}
It holds that $\bar y=0\Leftrightarrow y^*=0$ and otherwise $y^*/\bar y\in(0,1)$, since $\beta_1\mu_Y<r$ (obviously if $\mu_Y\leq 0$ and due to $\mathcal{Q}(r/\mu_Y)>0$ in \eqref{eq:quadratic} if $\mu_Y>0$); furthermore $\bar y>y_\cP$ iff $(c_0-c_A)/(c_B-c_A+rI)<\mu_Y/r$, e.g., if $c_0$ is sufficiently small or if $I$ or $c_B$ are sufficiently large (while $\mu_Y>0$).
}
Furthermore we show in the proof of the following proposition that it is always worthwhile to wait until $Y$ exceeds at least $y_\cP$ (if this is less than $y^*$). The continuation region $\cC$ and the stopping region $\cC^c$ for $\tilde L$ are then indeed separated by a function $b(x)$ as shown.

\begin{proposition}\label{prop:b(x)}
There exists a function $b:\R_+\to[\min(y_\cP,y^*),y^*]$ which is nondecreasing and continuous, such that, up to the origin, 
\begin{equation*}
\cC=\{(x,y)\in\R_+^2\mid\,y<b(x)\}\subset\cP^c
\end{equation*}
and $\tau_{\cC^c}:=\inf\{t\geq 0\mid(X_t,Y_t)\in\cC^c\}$ attains $V_{\tilde L}(x,y)=E\bigl[\tilde L_{\tau_{\cC^c}}\bigr]$ for $(X_0,Y_0)=(x,y)\in\R_+^2$. The origin belongs to $\cC$ iff $y^*>0$, i.e., iff $\bar y>0$. In that case $b$ further satisfies $b(x)\geq\min(\bar y,y_\cP+x(r-\mu_Y)/(r-\mu_X))$ (and otherwise $b\equiv y^*$ and $\cC=\emptyset$).
\end{proposition}

\noindent
{\it Proof:} See appendix.
\medskip

We show in Lemma \ref{lem:b<y*} below that indeed $b<y^*$ if $y_\cP<y^*$; nevertheless $b(x)\nearrow y^*$ since the value of the constrained problem converges to that of the unconstrained problem as $X_0\to\infty$.

By the shape of $\cC$ identified in Proposition \ref{prop:b(x)}, a firm that was sure never to become follower outside $\cP$ would switch as soon as the state leaves $\cC$: any delay on $\{Y>b(X)\}\cap\{F>L\}$ would induce a running expected loss given by the drift of $L$, which is $-e^{-rt}(Y_t-\bar y)\,dt<0$ there. Switching would yield the other, staying firm the prize $F>L$, however. Since both firms face the same situation, we obtain a war of attrition in the stopping region of the constrained stopping problem: there is a running cost of waiting for a prize that is obtained if the opponent gives in and switches. Therefore we call it the \emph{attrition region} 
\begin{equation*}
\cA:=\{(x,y)\in\R_+^2\mid y\geq b(x)\}\setminus\cP.\protect\footnotemark
\end{equation*}

\footnotetext{\label{fn:A=Cc}%
Formally we only have $\cA=\cC^c\cap\cP^c$ up to the origin by Proposition \ref{prop:b(x)}, but we prefer to work with the boundary representation. Precisely $\cA\supseteq\cC^c\cap\cP^c$, and as $b(x)$ lies below the preemption boundary $y_\cP+x(r-\mu_Y)/(r-\mu_X)$, in fact $b(0)=\min(y_\cP,y^*)$, so we have $(0,0)\in\cC\cap\cA$ iff $y^*>0\geq y_\cP$, resp.\ $\cA=\cC^c\cap\cP^c$ iff $y^*\leq 0$ or $y_\cP>0$.
}%
In order to derive equilibria outside the preemption region $\cP$, we need to know exactly the expected cost of abstaining to stop in $\cA$. In general this need not be just the (negative of the) drift of $L$, $e^{-rt}(Y_t-\bar y)\,dt\geq 0$ in $\cA$, since the state can transit very frequently between $\cA$ and $\cC$.\footnote{%
This ``switching on and off'' of the waiting cost could lead to non-trivial behavior, depending on the \emph{local time} our process $(X,Y)$ spends on the boundary between the two regions given by $y=b(x)$. See \citet{Jacka93} for an example based on Brownian motion where that local time is non-trivial.
}
Our equilibrium verification uses the following characterization of the constrained stopping problem and the cost of stopping too late. By the general theory of optimal stopping, the value \emph{process} of the stopping problem $V_{\tilde L}(X,Y):=U_{\tilde L}$ is the smallest supermartingale dominating the payoff process $\tilde L$, called the Snell envelope. As a supermartingale it has a decomposition $U_{\tilde L}=M_{\tilde L}-D_{\tilde L}$ with a martingale $M_{\tilde L}$ and a nondecreasing process $D_{\tilde L}$ called the \emph{compensator} starting from $D_{\tilde L}(0)=0$. Hence $U_{\tilde L}(0)-E[U_{\tilde L}(\tau)]=E[D_{\tilde L}(\tau)]\geq 0$ is the expected cost if one considers only stopping after $\tau$.

Given the geometry of the boundary between $\cA$ and $\cC$ that we have identified in Proposition \ref{prop:b(x)}, the increments of $D_{\tilde L}$ are indeed given by the (absolute value of the) drift of $L$ where the state is in $\cA$.

\begin{proposition}\label{prop:dD_Lswitch}
With $b(x)$ as in Proposition \ref{prop:b(x)} and $(X_0,Y_0)\not=(0,0)$ we have
\begin{equation}\label{dDtildeL}
dD_{\tilde L}(t)=\indi{t<\tau_\cP,Y_t\geq b(X_t)}e^{-rt}(Y_t-\bar y)\,dt
\end{equation}
for all $t\in\R_+$ a.s. If $(X_0,Y_0)=(0,0)$, \eqref{dDtildeL} still holds if $y^*\leq 0$ or $y_\cP>0$; otherwise $dD_{\tilde L}\equiv 0$.
\end{proposition}

\noindent
{\it Proof:} See appendix.
\medskip

With the concept of the Snell envelope $U_{\tilde L}$ we can now show that $b$ is not trivial, i.e., not just the constraint applied to the unconstrained solution.

\begin{lemma}\label{lem:b<y*}
If $y_\cP<y^*$, then also $b<y^*$.
\end{lemma}

\begin{proof}
Suppose $b(\hat x)=y^*>y_\cP$, hence $y^*>0$. Then in particular for all $X_0\geq\hat x$ and $Y_0=y^*$, $U_{\tilde L}(0)=\tilde L_0=L_{0}=U_L(0)$ with $U_L$ the (unconstrained) Snell envelope of $L$. For any stopping time $\tau$,
\begin{align*}
U_{\tilde L}(0)-E[U_{\tilde L}(\tau)]=E[D_{\tilde L}(\tau)]=E\biggl[\int_0^{\tau\wedge\tau_\cP}\indi{Y_t\geq b(X_t)}e^{-rt}(Y_t-\bar y)\,dt\biggr]
\end{align*}
and similarly for the unconstrained problem
\begin{align*}
U_{L}(0)-E[U_{L}(\tau)]=E[D_{L}(\tau)]=E\biggl[\int_0^{\tau}\indi{Y_t\geq y^*}e^{-rt}(Y_t-\bar y)\,dt\biggr].
\end{align*}
Let now $X_0>\hat x$, $Y_0=y^*$ and $\tau=\inf\{t\geq 0\mid X_t\leq\hat x\}\wedge\tau_\cP>0$. Then $U_{\tilde L}(\tau)<U_L(\tau)$ on $\{Y_\tau<y^*\}$, which has positive probability by our nondegeneracy assumption, so $E[U_{\tilde L}(\tau)]<E[U_{L}(\tau)]$.\footnote{%
For $Y_0>0$, the optimal stopping time for $L$, i.e., the first time $Y$ exceeds $y^*$, is unique (up to nullsets), but not admissible in the constrained problem for $Y_0,y_\cP<y^*$, so analogously $U_{\tilde L}(\tau)<U_L(\tau)$ on $\{Y_\tau<y^*\}$.
}
This contradicts all other terms in the previous two displays being equal, respectively.
\end{proof}

\begin{remark}\label{rem:b(x)}
It is a difficult problem to characterize $b$ more explicitly. However, by similar arguments as in the proof of Lemma \ref{lem:b<y*} one can obtain a scheme to approximate $b$ (see also our numerical study in Section \ref{sec:example}). The Snell envelope $U_{\tilde L}$, being a supermartingale of class (D), converges in $L^1(P)$ to $U_{\tilde L}(\infty)=\tilde L_\infty=L_{\tau_\cP}$. Also its martingale and monotone parts converge in $L^1(P)$ as $t\to\infty$, respectively. If $Y_0\geq b(X_0)$, we further have $U_{\tilde L}(0)=\tilde L_0=L_{0}$. Therefore,
\begin{align*}
U_{\tilde L}(0)-E[U_{\tilde L}(\infty)]&=L_{0}-E[L_{\tau_\cP}]=E\biggl[\int_0^{\tau_\cP}e^{-rt}(Y_t-\bar y)\,dt\biggr]\\
&=E[D_{\tilde L}(\infty)]=E\biggl[\int_0^{\tau_\cP}\indi{Y_t\geq b(X_t)}e^{-rt}(Y_t-\bar y)\,dt\biggr]
\end{align*}
and hence
\begin{equation}\label{eq:b(x)}
Y_0\geq b(X_0)\quad\Rightarrow\quad E\biggl[\int_0^{\tau_\cP}\indi{Y_t<b(X_t)}e^{-rt}(Y_t-\bar y)\,dt\biggr]=0.
\end{equation}
Since $y^*\geq b$, one can also use $\tau_\cP\wedge\inf\{t\geq 0\mid Y_t\geq y^*\}$ in \eqref{eq:b(x)}. Any candidate $\hat b$ for $b$ should satisfy $\hat b\geq \bar y$ like $b$ itself. Therefore, if $\hat b$ is (locally) too high (low), the expectation in \eqref{eq:b(x)} will be positive (negative) for $Y_0=\hat b(X_0)$.
\end{remark}

%\begin{remark}\label{rem:b(x)}
%The Snell envelope $U_{\tilde L}$, being a supermartingale of class (D), converges in $L^1(P)$ to $U_{\tilde L}(\infty)=\tilde L_\infty=L_{\tau_\cP}$. Also its martingale and monotone parts converge in $L^1(P)$ as $t\to\infty$, respectively. If $Y_0=b(X_0)$ we further have $U_{\tilde L}(0)=\tilde L_0=L_{0}$. Therefore,
%\begin{align*}
%&U_{\tilde L}(0)-E[U_{\tilde L}(\infty)]=L_{0}-E[L_{\tau_\cP}]\\
%&=\frac{b(X_0)}{r-\mu_Y}-\frac{c_B-c_0}{r}-I-E\biggl[e^{-r\tau_\cP}\biggl(\frac{X_{\tau_\cP}}{r-\mu_X}-\frac{c_A-c_0}{r}\biggr)\biggr]\\
%&=E[D_{\tilde L}(\infty)]\\
%&=E\biggl[\int_0^{\tau_\cP}\indi{Y_t\geq b(X_t)}e^{-rt}(Y_t-c_B+c_0-rI)\,dt\biggr].
%\end{align*}
%If we denote the solution to \eqref{SDE} for $X_0=Y_0=1$ by $(X^1,Y^1)$, we have for arbitrary $(X_0,Y_0)=(x,b(x))\in\R_+^2$
%\begin{align*}
%&\frac{b(x)}{r-\mu_Y}-\frac{c_B-c_0}{r}-I-E\biggl[e^{-r\tau_\cP(x,b(x))}\biggl(\frac{xX^1_{\tau_\cP(x,b(x))}}{r-\mu_X}-\frac{c_A-c_0}{r}\biggr)\biggr]\\
%&=E\biggl[\int_0^{\tau_\cP(x,b(x))}\indi{b(x)Y^1_t\geq b(xX^1_t)}e^{-rt}(b(x)Y^1_t-c_B+c_0-rI)\,dt\biggr],
%\end{align*}
%where we indicate the dependence of the hitting time of $\cP$ on the initial point by $\tau_\cP(x,b(x))$. This can be understood as an integral-type equation for determining $b(\cdot)$.
%\end{remark}

\subsection{Subgame Perfect Equilibria}\label{subsec:eql}

We are now ready to combine the previous results to construct equilibria for arbitrary subgames. If the state reaches $\cP$, there is preemption with immediate switching by some firm. In the continuation region $\cC$, the firms just wait. In the attrition region $\cA$, however, the firms switch at a specific rate such that the resulting probability for the other to become (more profitable) follower compensates for the negative drift of the leader's payoff, making each indifferent to wait or switch immediately.

\begin{proposition}\label{prop:eqlF>L}
Suppose $(X_0,Y_0)\not=(0,0)$ or $y_\cP>0$. Fix $\vartheta\in\T$ and set $\tau_\cP:=\inf\{t\geq\vartheta\mid L_t>F_t\}$. Then $\bigl(G^\vartheta_1,\alpha^\vartheta_1\bigr)$, $\bigl(G^\vartheta_2,\alpha^\vartheta_2\bigr)$ defined by
\begin{equation}\label{eq:eqlF>L}
\frac{dG^\vartheta_i(t)}{1-G^\vartheta_i(t)}=\frac{\indi{Y_t\geq b(X_t)}(Y_t-\bar y)\,dt}{X_t/(r-\mu_X)-(Y_t-y_\cP)/(r-\mu_Y)}
\end{equation}
on $[\vartheta,\tau_\cP)$ and $G^\vartheta_i(t)=1$ on $[\tau_\cP,\infty]$ and
\begin{equation*}
\alpha^\vartheta_i(t)=\indi{L_t>F_t}\frac{L_t-F_t}{L_t-M_t}
\end{equation*}
on $[\vartheta,\infty]$, $i=1,2$, are an equilibrium in the subgame at $\vartheta$.

The resulting payoffs are
\begin{equation*}
V_i^\vartheta(G_i^\vartheta,\alpha^\vartheta_i,G_j^\vartheta,\alpha^\vartheta_j)=V_{\tilde L}(X_\vartheta,Y_\vartheta)=\esssup_{\tau\geq\vartheta}E\bigl[\tilde L_\tau\bigv\F_\vartheta\bigr]
\end{equation*}
with $\tilde L:=L\indi{t<\tau_\cP}+F_{\tau_\cP}\indi{t\geq\tau_\cP}$.
\end{proposition}

\begin{proof}
Fix $i\in\{1,2\}$ and let $j$ denote the other firm. The strategies satisfy the conditions of Definition \ref{def:alpha} (presupposing $G^\vartheta_i(t)=\alpha^\vartheta_i(t)=0$ on $[0,\vartheta)$). Note in particular that by our nondegeneracy assumption, $\tau_\cP=\inf\{t\geq\vartheta\mid L_t\geq F_t\}$ unless $(X_\vartheta,Y_\vartheta)=(X_0,Y_0)=(0,0)$ and $y_\cP=0$, and thus except for that excluded case the denominator in \eqref{eq:eqlF>L} will not vanish before $\tau_\cP$. Now $\bar y>Y_t\geq b(X_t)$ implies $t\geq\tau_\cP$ by $b(x)\geq\min(\bar y,y_\cP+x(r-\mu_Y)/(r-\mu_X))$ and hence the proposed stopping rate is also nonnegative.\footnote{%
Given the proposed rate, $G^\vartheta_i(t)=1-\exp\{-\int_\vartheta^t(1-G^\vartheta_i(s))^{-1}\,dG^\vartheta_i(s)\}$ on $[\vartheta,\tau_\cP)$.
}

The strategies are mutual best replies at $\tau_\cP$ by Proposition \ref{prop:eqlL>F}, i.e., given any value of $G^\vartheta_i(\tau_\cP-)$, $G^\vartheta_i(t)=1$ on $[\tau_\cP,\infty]$ and $\alpha^\vartheta_i(t)$ as proposed are optimal against $\bigl(G^\vartheta_j,\alpha^\vartheta_j\bigr)$ and the related payoff to firm $i$ is $E\bigl[(1-G^\vartheta_i(\tau_\cP-)\bigr)\bigl(1-G^\vartheta_j(\tau_\cP-)\bigr)F_{\tau_\cP}\bigv\F_\vartheta\bigr]$. It remains to show optimality of $G^\vartheta_i$ on $[\vartheta,\tau_\cP)$. Since $G^\vartheta_j$ is continuous on $[\vartheta,\tau_\cP)$ it follows from \cite{RiedelSteg14} that it is not necessary to consider the possibility that $\alpha^\vartheta_i(t)>0$.

Given the proposed optimal $\alpha^\vartheta_i$ and $G^\vartheta_i(t)=1$ on $[\tau_\cP,\infty]$, the related payoff at $\tau_\cP$ and the continuity of $G^\vartheta_j$ up to $\tau_\cP$, we can write the payoff from any $G^\vartheta_i$ on $[\vartheta,\tau_\cP)$ as
\begin{align*}%\label{Vi_mixed}
V^\vartheta_i\bigl(G^\vartheta_i,\alpha^\vartheta_i,G^\vartheta_j,\alpha^\vartheta_j\bigr)=E&\biggl[\int_{[0,\tau_\cP]}\biggl(\int_{[0,s)}F_t\,dG^\vartheta_j(t)+\bigl(1-G^\vartheta_j(s)\bigr)L_s\biggr)\,dG^\vartheta_i(s)\\
&+(1-G^\vartheta_i(\tau_\cP-)\bigr)\bigl(1-G^\vartheta_j(\tau_\cP-)\bigr)F_{\tau_\cP}\biggv\F_\vartheta\biggr].
\end{align*}
Further, noting $(1-G^\vartheta_i(\tau_\cP-)\bigr)\bigl(1-G^\vartheta_j(\tau_\cP-)\bigr)F_{\tau_\cP}=\Delta G^\vartheta_i(\tau_\cP)\Delta G^\vartheta_j(\tau_\cP)F_{\tau_\cP}$ and $\Delta G^\vartheta_j(t)=0$ on $[0,\tau_\cP)$, we obtain that
\begin{align*}%\label{Vi_mixed}
V^\vartheta_i\bigl(G^\vartheta_i,\alpha^\vartheta_i,G^\vartheta_j,\alpha^\vartheta_j\bigr)=E&\biggl[\int_{[0,\tau_\cP]}S_i(s)\,dG^\vartheta_i(s)\biggv\F_\vartheta\biggr]
\end{align*}
where $S_i(s)=\int_{[0,s]}F_t\,dG^\vartheta_j(t)+\bigl(1-G^\vartheta_j(s)\bigr)L_s$. Thus $G^\vartheta_i$ is optimal iff it increases when it is optimal to stop $S_i$. We now show that this is the case anywhere in the attrition region $\cA$ or at $\tau_\cP$.

The process $S_i$ is continuous on $\{\vartheta<\tau_\cP\}$ because there $L_{\tau_\cP}=F_{\tau_\cP}$ implies $\Delta S_i(\tau_\cP)=\Delta G^\vartheta_j(\tau_\cP)\bigl(F_{\tau_\cP}-L_{\tau_\cP}\bigr)=0$. Since $L$ is a continuous semimartingale and $G^\vartheta_j$ is monotone and continuous on $[0,\tau_\cP)$, $S_i$ satisfies
\begin{equation*}
dS_i(s)=(F_s-L_s)\,dG^\vartheta_j(s)+\bigl(1-G^\vartheta_j(s)\bigr)\,dL_s.
\end{equation*}
Using $(F_s-L_s)dG^\vartheta_j(s)=\bigl(1-G^\vartheta_j(s)\bigr)\indi{Y_s\geq b(X_s)}e^{-rs}(Y_s-\bar y)\,ds$ on $[\vartheta,\tau_\cP)$ yields
\begin{equation*}
dS_i(s)=\bigl(1-G^\vartheta_j(s)\bigr)\bigl(dL_s+\indi{Y_s\geq b(X_s)}e^{-rs}(Y_s-\bar y)\,ds\bigr).
\end{equation*}
Recalling $\tilde L$ from Subsection \ref{subsec:stop} (with $\tau_\cP$ adjusted for $\vartheta$), and its Snell envelope $U_{\tilde L}$ and compensator $D_{\tilde L}$ identified in Proposition \ref{prop:dD_Lswitch}, we have in fact
\begin{equation*}
dS_i(s)=\bigl(1-G^\vartheta_j(s)\bigr)\bigl(d\tilde L_s+dD_{\tilde L}(s)\bigr)
\end{equation*}
on $[\vartheta,\tau_\cP]$ where $\vartheta<\tau_\cP$. We first ignore $(1-G^\vartheta_j)$ and argue that it is optimal to stop $\tilde L+D_{\tilde L}$ anywhere in the attrition region $\cA$ or at $\tau_\cP$. Indeed, as $U_{\tilde L}=M_{\tilde L}-D_{\tilde L}$ is the smallest supermartingale dominating $\tilde L$, $U_{\tilde L}+D_{\tilde L}=M_{\tilde L}$ must be the smallest supermartingale dominating $\tilde L+D_{\tilde L}$, i.e., the Snell envelope of $\tilde L+D_{\tilde L}$. As $M_{\tilde L}$ is a martingale, any stopping time $\tau$ with $M_{\tilde L}(\tau)=\tilde L_\tau+D_{\tilde L}(\tau)\Leftrightarrow U_{\tilde L}(\tau)=\tilde L_\tau$ a.s.\ is optimal for $\tilde L+D_{\tilde L}$, which is satisfied where the proposed $G^\vartheta_i$ increases (i.e., in $\cA$ or at $\tau_\cP$).\footnote{%
By Proposition \ref{prop:b(x)}, $\cA\cap\cC$ could only be reached at the origin, i.e., by $(X_\tau,Y_\tau)=(0,0)\Leftrightarrow(X_0,Y_0)=(0,0)$ and if $y_\cP=0$ (cf.\ footnote \ref{fn:A=Cc}), which we have excluded. $U_{\tilde L}(\tau_\cP)=\tilde L_{\tau_\cP}$ holds by definition.
}
Instead of stopping $\tilde L+D_{\tilde L}$, we actually have to consider, for any $\tau\in[\vartheta,\tau_\cP]$,
\begin{align*}
S_i(\tau)-S_i(\vartheta)&=\int_{[\vartheta,\tau]}\bigl(1-G^\vartheta_j(s)\bigr)\bigl(d\tilde L_s+dD_{\tilde L}(s)\bigr)\\
&=\bigl(1-G^\vartheta_j(\tau)\bigr)\bigl(\tilde L_\tau+D_{\tilde L}(\tau)\bigr)-\tilde L_\vartheta-D_{\tilde L}(\vartheta)\\
&+\int_{[\vartheta,\tau]}\bigl(\tilde L_s+D_{\tilde L}(s)\bigr)\,dG^\vartheta_j(s).
\end{align*}
The presence of $(1-G^\vartheta_j)$ thus acts like a constraint compared to stopping $\tilde L+D_{\tilde L}$. However, $G^\vartheta_j$ places its entire mass favorably, as it also only increases when it is optimal to stop $\tilde L+D_{\tilde L}$. Therefore the payoff from optimally stopping $S_i$ is indeed that from $\tilde L+D_{\tilde L}$, i.e., $M_{\tilde L}(\vartheta)=U_{\tilde L}(\vartheta)$ (with $D_{\tilde L}(\vartheta)=0$), and the same stopping times are optimal. The proof is now complete.
\end{proof}

The Markovian equilibrium stopping rate \eqref{eq:eqlF>L} is the (absolute value of the) drift of $L$ in $\cA$ divided by the difference $F-L$, ensuring the exact compensation that makes each firm indifferent as discussed before. Waiting is strictly optimal if $(X,Y)\in\cC$, where the equilibrium payoff is $U_{\tilde L}>L$.

Although the stopping rate \eqref{eq:eqlF>L} becomes unboundedly large as the state approaches the preemption region~-- because the denominator $F-L$ vanishes while the numerator is bounded away from zero~-- the cumulative stopping probability does \emph{not} converge to 1. There will indeed be some mass left when reaching the preemption boundary. This is a distinctive feature of our stochastic model\footnote{%
The mathematical question underlying Proposition \ref{prop:DG>0} is of interest in its own right. With the same arguments used in the proof (which then become much simpler) one can show that for a Brownian motion $B$ started at $a>0$ one has
\begin{equation*}
\int_0^{\tau_0}\frac{1}{B_t}\,dt<\infty
\end{equation*}
a.s., where $\tau_0=\inf\{t\geq 0\mid B_t\leq 0\}$. We actually consider the reciprocal of the process $Z_t=X_t-Y_t+a$ with our geometric Brownian motions $X$, $Y$.
}
that is not observed in deterministic versions, cf.\ Section \ref{sec:dtm}.

\begin{proposition}\label{prop:DG>0}
The strategies $G^\vartheta_i$ specified in Proposition \ref{prop:eqlF>L} satisfy $\Delta G^\vartheta_i(\tau_\cP)>0$ on $\{\tau_\cP<\infty\}$ a.s.
\end{proposition}

\noindent
{\it Proof:} See appendix.
\medskip

Concerning the regularity of $\alpha_i$, note that we can only have $L\leq M$ on $\{Y\leq y_\cP\}$ by our assumption $(c_0-c_A)/r+I\geq 0$, whence $L>M$ on $\{(X,Y)\in\cP\cup\partial\cP\}$ a.s.\ if $X_0>0$; so $\alpha_i$ will be continuous and vanish on $\{(X,Y)\in\partial\cP\}$.\footnote{%
In the particular case $X_0=0$, $X\equiv 0$ and
\begin{equation*}
\alpha_i(t)=\frac{L_t-F_t}{L_t-M_t}\indi{L_t>F_t}=\frac{(Y_t-y_\cP)\indi{Y_t>y_\cP}}{Y_t-y_\cP+(r-\mu_Y)(I-(c_A-c_0)/r)}.
\end{equation*}
Then $\alpha_i$ will still be continuous for $I>(c_A-c_0)/r$ and unity on $\{Y\geq y_\cP\}$ for $I=(c_A-c_0)/r$.
}

We can easily aggregate the equilibria to a subgame perfect equilibrium since all quantities are Markovian.

\begin{theorem}\label{thm:eql}
The strategies $\bigl(G^\vartheta_i,\alpha^\vartheta_i\bigr)_{\vartheta\in\T}$, $i=1,2$ of Proposition \ref{prop:eqlF>L} form a subgame perfect equilibrium.
\end{theorem}

\begin{remark}\label{rem:Y=0}
If $X_0=0<Y_0$, the game is played ``on the \nbd{y}axis'' and the derived equilibria are as follows. There is preemption as soon as $Y$ exceeds the preemption point $y_\cP$. Any $y$ less than $b(0)=\min(y_\cP,y^*)$ is in the continuation region, so there is attrition iff $(y^*)^+<y_\cP$, between those points.

If $Y_0=0$, the game is played ``on the \nbd{x}axis'' and is actually deterministic since we then have $\tau_\cP=\infty$ by our assumption $y_\cP\geq 0$. In fact, now $F_t>L_t\:\forall t\in\R_+$ if $X_0\vee y_\cP>0$, and $dL_t=e^{-rt}\bar y\,dt$. Thus, if $\bar y\geq 0$, stopping is dominated for $t\in\R_+$, strictly if $\bar y>0$ or against any strategy that charges $[0,\infty)$, while the proposed equilibrium strategies only charge $[\infty]$.\footnote{%
Indeed, if $\bar y=y^*=0$, $b\equiv y^*=0$, while if $\bar y>0\Leftrightarrow y^*>0$, then $b(X_t)>0$ for $X_0+y_\cP>0$ by $b(x)\geq\min(\bar y,y_\cP+x(r-\mu_Y)/(r-\mu_X))$. 
}
If $\bar y<0$, also $b\equiv y^*<0$, and there is permanent attrition with rate $-dL/(F-L)$.

In the completely degenerate case that further $X_0=y_\cP=0$, the stopping rate of Proposition \ref{prop:eqlF>L} is not well defined. Now, however, $L\equiv F$ and there are the following equilibria: for $\bar y>0$, both firms wait indefinitely, for $\bar y=0$, any strategy pair with no joint mass points are an equilibrium, and for $\bar y<0$, one firm switches immediately and the other abstains.
\end{remark}

\section{A Numerical Example}\label{sec:example}

As an illustration of the model we present a numerical example to show the consequences of the results presented so far. The parameter values of our base case are given in Table~\ref{tab:values}.
\begin{table}[h]
\centering
\begin{tabular}{ccccc}
  \toprule
  $c_0=3.5$ & $c_A=4.5$ & $c_B=5$ & $I=100$ & $r=.1$ \\
  $\mu_Y=.04$ & $\mu_X=.04$ & $\sigma_Y=.25$ & $\sigma_X=.25$ & $\rho=.4$\\
  \bottomrule
\end{tabular}
\caption{Parameter values for a numerical example}
\label{tab:values}
% Created with HTAposterior.m
\end{table}

For this case it holds that $y^{\ast}=17.45$. The preemption region is depicted in Figure~\ref{fig:preempregion}. Finding the boundary $x\mapsto b(x)$ as characterized in Proposition~\ref{prop:b(x)} is not trivial as can be seen from Remark~\ref{rem:b(x)}. However, it is known from the literature on numerical methods for pricing American options (see, for example, \citealp{Glass04} for an overview) that a piecewise linear or exponential approximation often works well. In our case, we hypothesize an exponential boundary of the form
\begin{equation*}
  b(x) = y^{\ast}-(y^{\ast}-\bar{y})e^{-\gamma(x-\bar{x})},
\end{equation*}
where $\bar{x}$ is the unique such that $(\bar{x},\bar{y})$ is on the boundary of $\cP$, and $\gamma$ is a free parameter.

A rough implementation now consists of hypothesizing a value for $\gamma$, followed by simulating 3,000 sample paths for several starting values close to this boundary. The optimal boundary should be such that simulated continuation values are uniformly slightly larger than the leader value. Using a manual search across several values for $\gamma$ suggests a value of .0984 and a boundary as depicted in Figure~\ref{fig:ConstrBd}.

\begin{figure}[htb]
  \centering
  \subfloat[Preemption region]{\includegraphics[height=6cm,width=5cm]{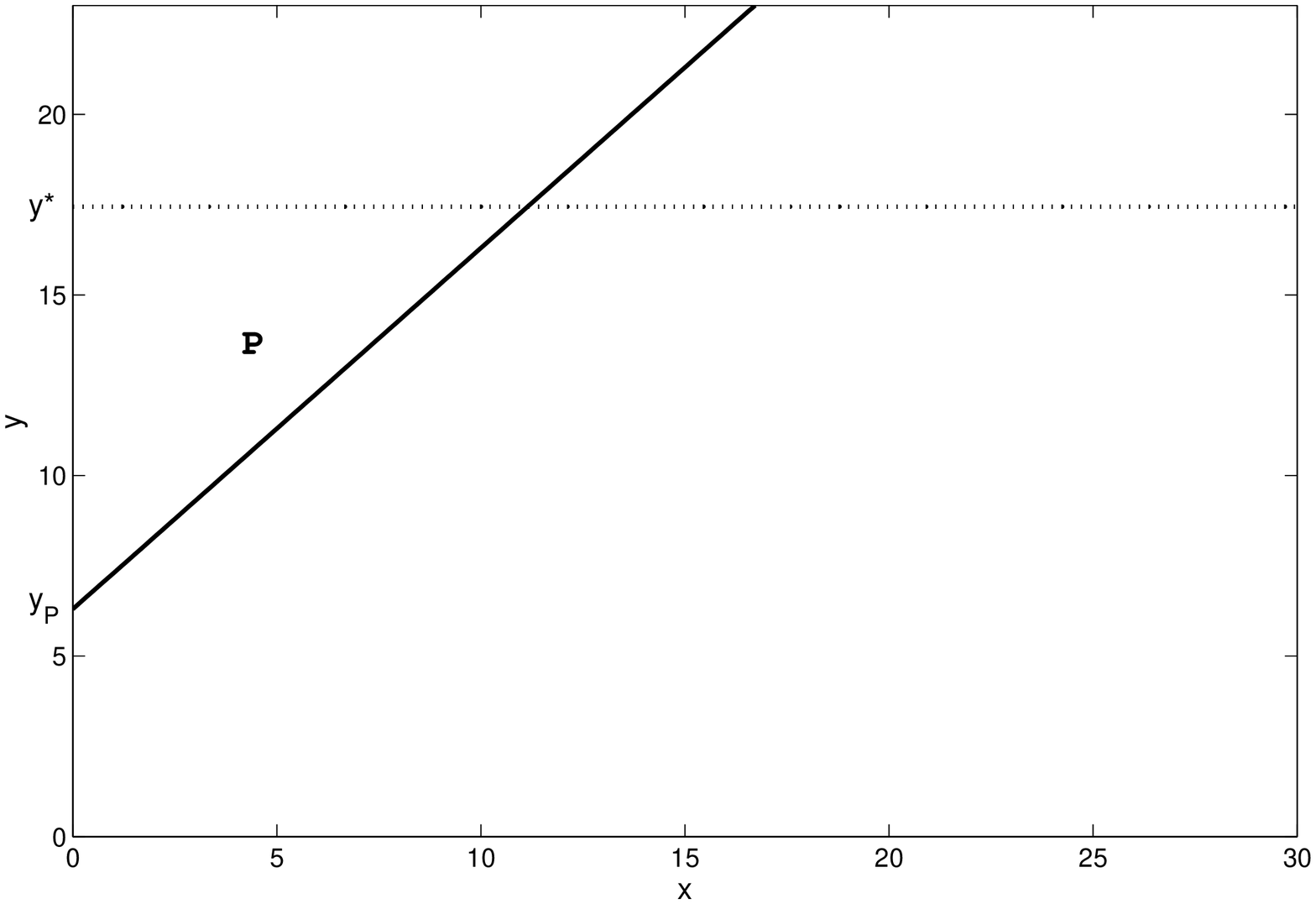}\label{fig:preempregion}}\qquad
  \subfloat[Optimal stopping boundary for constrained leader value]{\includegraphics[height=6cm,width=5cm]{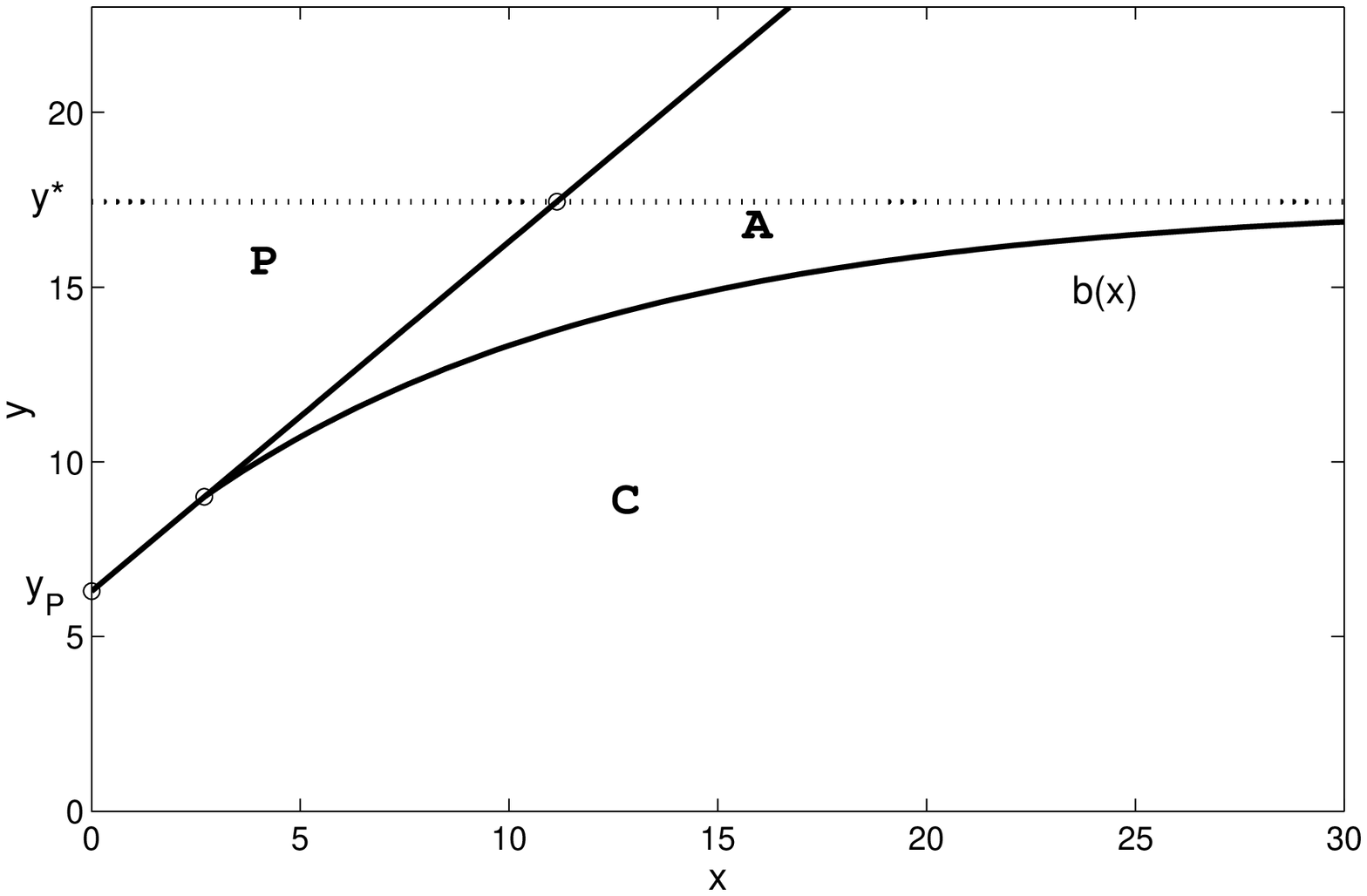}\label{fig:ConstrBd}}
  \caption{Preemption, attrition and continuation regions for numerical example.}
  % Created with coffeeExp.m
\end{figure}

Once the boundary is established we can simulate equilibrium sample paths using the equilibrium strategies from Theorem~\ref{thm:eql}. Two sample paths, together with their realizations of the attrition rate~\eqref{eq:eqlF>L} are given in Figures~\ref{fig:Example1} and~\ref{fig:Example2}. In parallel to the paths we sample the investment decisions from the respective attrition rates. The displayed paths end where investment takes place.
\begin{figure}[htb]
  \centering
  \includegraphics[height=6cm,width=12cm]{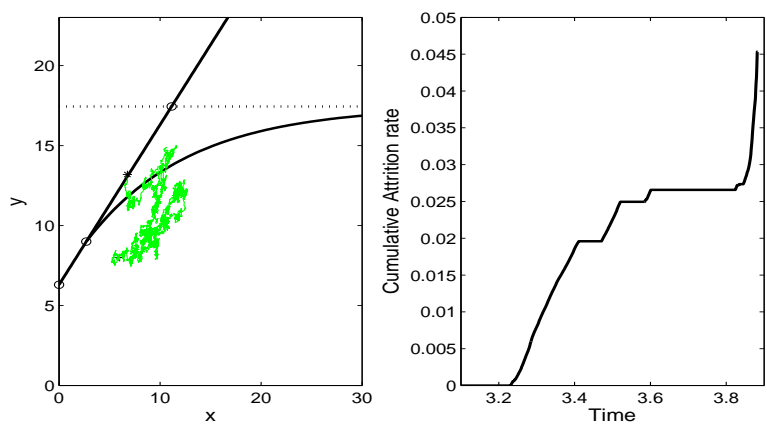}
  \caption{Sample path ending in preemption and its realized attrition rate.}
  \label{fig:Example1}
  % Created with coffeeExp.m
\end{figure}

\begin{figure}[htb]
  \centering
  \includegraphics[height=6cm,width=12cm]{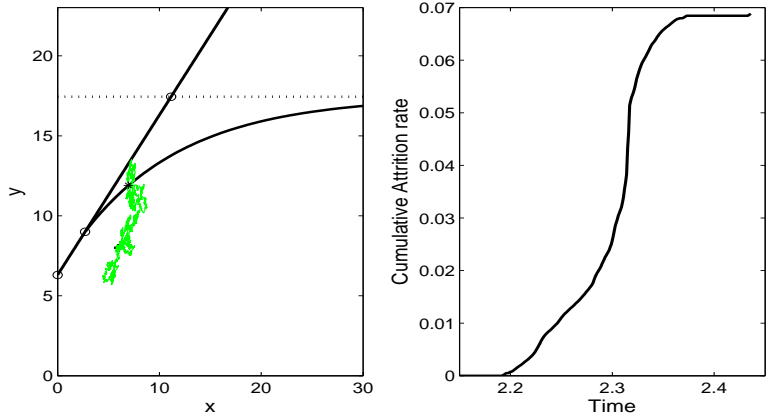}
  \caption{Sample path ending in attrition region and its realized attrition rate.}
  \label{fig:Example2}
  % Created with coffeeExp.m
\end{figure}

Note that, even though the two sample paths start at exactly the same point $(X_0,Y_0)=(6,8)$, their equilibrium realizations are very different. In the sample path $t\mapsto(X_t(\omega_1),Y_t(\omega_1))$ in Figure~\ref{fig:Example1} the attrition region is entered and exited several times before, finally, the sample path hits the preemption boundary. At that time the realized attrition rate has gone up to about $G^0_j(\tau_\cP(\omega_1)-)\approx .045$, $j=1,2$. However, no firm has invested up to time $\tau_\cP(\omega_1)$. At time $\tau_\cP(\omega_1)$ we have $(G^{\tau_\cP(\omega_1)}_j(\tau_\cP(\omega_1)),\alpha^{\tau_\cP(\omega_1)}_j(\tau_\cP(\omega_1)))=(1,0)$, $j=1,2$, so that (see Appendix~\ref{app:outcome}) each firm is the first to invest at time $\tau_\cP(\omega_1)$ with probability $\lambda^{\tau_\cP(\omega_1)}_{L,i}=.5$. So, the $\omega_1$ sample path ends with investment taking place when there is, just, a first-mover advantage. In the sample path $t\mapsto(X_t(\omega_2),Y_t(\omega_2))$ in Figure~\ref{fig:Example2} the attrition region is again entered and exited several times. Note that the first entry of the attrition region is at roughly the same time as along the $\omega_1$ sample path. At time 2.5 (approximately), when the attrition rate realization has gone up to $G^0_j(1.4)\approx .07$, $j=1,2$, one of the two firms actually invests, thereby handing the other firm the second-mover advantage.

Finally, we run a simulation of 3,000 sample paths, all starting at $(X_0,Y_0)=(6,8)$, and record whether the game stops in the attrition or preemption region. We find that 80\% of sample paths end with preemption, whereas 20\% of sample paths end in the attrition region. A scatter diagram of values $(X_{\tau(\omega)},Y_{\tau(\omega)})$ at the time of first investment is given in Figure~\ref{fig:scatter}.\footnote{This figure excludes outliers in both dimensions.} The average fraction of time spent in the attrition region is 6.76\% and all of sample paths experience attrition before the time of first investment. The value of immediate investment is -16.67, whereas the (simulated) value of the equilibrium strategies is 12.02.
\begin{figure}[htb]
  \centering
  \includegraphics[height=7cm,width=7cm]{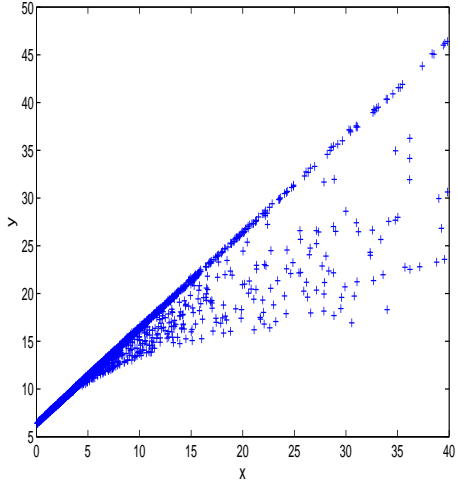}
  \caption{Scatter plot of realizations of $(X_{\tau^{\ast}},Y_{\tau^{\ast}})$.}
  \label{fig:scatter}
  % Created with coffeePLpath.m
\end{figure}

\section{Differences between Deterministic and Stochastic Models}\label{sec:dtm}

For comparative purposes, we now briefly analyse a deterministic version of our model, i.e., setting $\sigma_X=\sigma_Y=0$. The expressions for $L_t$, $F_t$ and $M_t$ provided in \eqref{L(y)} then remain the same. In order to avoid the distinction of many different cases, we concentrate on the most ``typical'' one representing the economic story that we have in mind. Therefore assume that the market $B$ is \emph{growing} at the rate $\mu_Y>0$ from the initial value $Y_0>0$. Further assume that $\mu_Y\geq\mu_X$, so that it is only market $B$ that can possibly overtake the other, $A$, in profitability.\footnote{%
If $\mu_X>\mu_Y$ and $X_0>0$, the behavior differs from that discussed in the main text as follows. Now the profitability of market $A$ will eventually outgrow that of market $B$ and preemption occurs during a bounded time interval (if at all). Indeed, since
\begin{equation}\label{preemp_dtm2}
\frac{\partial(L_t-F_t)}{\partial t}=e^{\mu_Xt}\biggl(\frac{\mu_YY_0}{r-\mu_Y}e^{(\mu_Y-\mu_X)t}-\frac{\mu_XX_0}{r-\mu_X}\biggr)
\end{equation}
has at most one zero and is negative for $t$ sufficiently large, $(L_t-F_t)$ grows until its maximum and then declines, so $\{t\mid L_t-F_t>0\}$ forms an interval (if nonempty). Arranging $(L_t-F_t)$ similarly as \eqref{preemp_dtm2} shows that $(L_t-F_t)$ is negative for $t$ sufficiently large, which bounds the previous interval.

Supposing that the preemption interval is nonempty and that $\mu_Y>0$, the unique maximum of $L_t$ attained at $t=\bar t$ may occur before, during or after preemption. We can observe attrition followed by preemption only in the first case, i.e., if $L_{\bar t}<F_{\bar t}$ and $\partial(L_{\bar t}-F_{\bar t})>0$, which is if $X_{\bar t}/(r-\mu_X)$ (see the right of \eqref{X_tbar}) lies in $\bigl(\bar y/(r-\mu_Y)-(c_B-c_A)/r-I,\mu_Y\bar y/(\mu_X(r-\mu_Y))\bigr)$. This interval is nonempty, e.g., if $\mu_X$ and $\mu_Y$ are sufficiently close.

In subgames after the preemption interval there will be indefinite attrition for all $t\geq\bar t$.
}
The same holds then for the functions $L$ and $F$: Now we have preemption due to $L_t>F_t$ iff $(X_t,Y_t)\in\cP$, which is iff
\begin{equation}\label{preemp_dtm}
e^{\mu_Yt}\biggl(\frac{Y_0}{r-\mu_Y}-\frac{X_0}{r-\mu_X}e^{(\mu_X-\mu_Y)t}\biggr)>\frac{c_B-c_A}{r}+I.
\end{equation}
The LHS is the product of two nondecreasing functions and the RHS is nonnegative by our assumption made in Section \ref{sec:model}. Hence $L_t>F_t$ for all
\begin{equation*}
t>\inf\{t\geq 0\mid L_t>F_t\}=:t_\cP
\end{equation*}
and preemption will indeed occur iff the LHS of \eqref{preemp_dtm} ever becomes positive, i.e., $t_\cP<\infty\,\Leftrightarrow\,\mu_Y>\mu_X$ or $Y_0/(r-\mu_Y)>X_0/(r-\mu_X)$. Preemption does not start immediately at $t=0$ iff  $Y_0/(r-\mu_Y)-X_0/(r-\mu_X)<(c_B-c_A)/r+I$, i.e., $t_\cP>0\,\Leftrightarrow\,(X_0,Y_0)\not\in\overline\cP$.

Now the constrained stopping problem to maximize $L_t$ over $t\leq t_\cP$ is very easy. As $L_t$ is increasing iff $Y_t<\bar y$, i.e., iff $Y_0<\bar y$ and
\begin{equation*}
t<\frac{1}{\mu_Y}\ln\biggl(\frac{\bar y}{Y_0}\biggr)=:\bar t,
\end{equation*}
the solution is to stop at the unconstrained optimum $t=\bar t$ (resp.\ at $t=\bar t:=0$ if $Y_0\geq\bar y$) or at the constraint $t=t_\cP$ if it is binding. The solution is also unique since $L_t$ in fact strictly increasing for $Y_t<\bar y$ and strictly decreasing for $Y_t\geq\bar y$.

As a consequence, stopping is strictly dominated before the constrained optimum $\min(\bar t,t_\cP)$ is reached and we can only observe attrition in equilibrium if the \emph{un}constrained maximum of $L$ is reached \emph{before} preemption starts. It indeed holds that $L_{\bar t}<F_{\bar t}$ (given $Y_0<\bar y$) iff
\begin{flalign}\label{X_tbar}
&& \frac{Y_{\bar t}}{r-\mu_Y}-\frac{c_B-c_A}{r}-I&<\frac{X_{\bar t}}{r-\mu_X} &&\nonumber\\
&\Leftrightarrow & \frac{\bar y}{r-\mu_Y}-\frac{c_B-c_A}{r}-I&<\frac{X_0e^{\mu_X\bar t}}{r-\mu_X}=\frac{X_0}{r-\mu_X}\biggl(\frac{\bar y}{Y_0}\biggr)^\frac{\mu_X}{\mu_Y} &&
\end{flalign}
(and $L_0<F_0\,\Leftrightarrow(X_0,Y_0)\not\in\overline\cP$ for $Y_0\geq\bar y$), so that there will be attrition iff the initial profitability $X_0$ of market $A$ is large enough to let $Y_t$ exceed $\bar y$ before $(X_t,Y_t)\in\overline\cP$.\footnote{%
If $\mu_X\not=0$, the path that $(X_t,Y_t)$ takes in the state space $R_+^2$ is given by the relation $y=Y_0(x/X_0)^{(\mu_Y/\mu_X)}$. For $\mu_X=0$, it is of course just $\{(X_0,y)\mid y\geq Y_0\}$.
}
Once $Y_t>\bar y$, resp.\ $t>\bar t$, $L_t$ strictly decreases and we observe attrition until preemption starts at $t_\cP$.

The equilibrium can now be represented as in Proposition \ref{prop:eqlF>L}, resp.\ Theorem \ref{thm:eql}, with the fixed preemption start date $t_\cP$ replacing $\tau_\cP(\vartheta)$ throughout and ${t\geq\bar t}$ replacing the dynamic attrition boundary ${Y_t\geq b(X_t)}$. Note that if there is attrition before the preemption region is reached, stopping will now occur with probability 1 due to attrition.\footnote{%
In analogy (in fact contrast) to Proposition \ref{prop:DG>0} and its proof it suffices to verify that
\begin{equation*}
\int_0^{\tau_\cP}\frac{dt}{X_t-Y_t+a}=\infty
\end{equation*}
when in the simplified notation $\tau_\cP=\inf\{t\geq 0\mid X_t-Y_t\leq a\}\in(0,\infty)$, since we are only interested in the case $Y_{\tau_\cP}>\bar y$ so that $Y_t$ is bounded away from $\bar y$ near $\tau_\cP$, where the above integral explodes. Indeed, with $Z_t=X_t-Y_t+a$ we have $\lim_{t\to\tau_\cP}Z_t=0$ and $dZ_t=(\mu_XX_0e^{\mu_Xt}-\mu_YY_0e^{\mu_Yt})\,dt$ and hence
\begin{equation*}
\infty=\lim_{t\to\tau_\cP}-\ln(Z_t)=-\ln(Z_0)-\int_0^{\tau_\cP}\frac{\mu_XX_0e^{\mu_Xt}-\mu_YY_0e^{\mu_Yt}}{X_t-Y_t+a}\,dt.
\end{equation*}
As the numerator in the integral is bounded on the assumed finite interval $[0,\tau_\cP]$, the claim follows.
}

As a numerical example, we use the same parameter values as before, but set $\sigma_X=\sigma_Y=0$. For easy reference, the left-panel of Figure~\ref{fig:dtm} includes the approximate boundary for the attrition region as derived in Section~\ref{sec:example}. Note that for each sample path it is clear \emph{a priori} that the game ends in attrition and that, once the attrition region is entered, it will never be exited again. As a result, the equilibrium attrition rate is a smooth function of time as can be seen in the right-panel of Figure~\ref{fig:dtm}. In this example non of the sample paths will ever end in preemption.

\begin{figure}[htb]
  \centering
  \subfloat[Sample path]{\includegraphics[height=6cm,width=5cm]{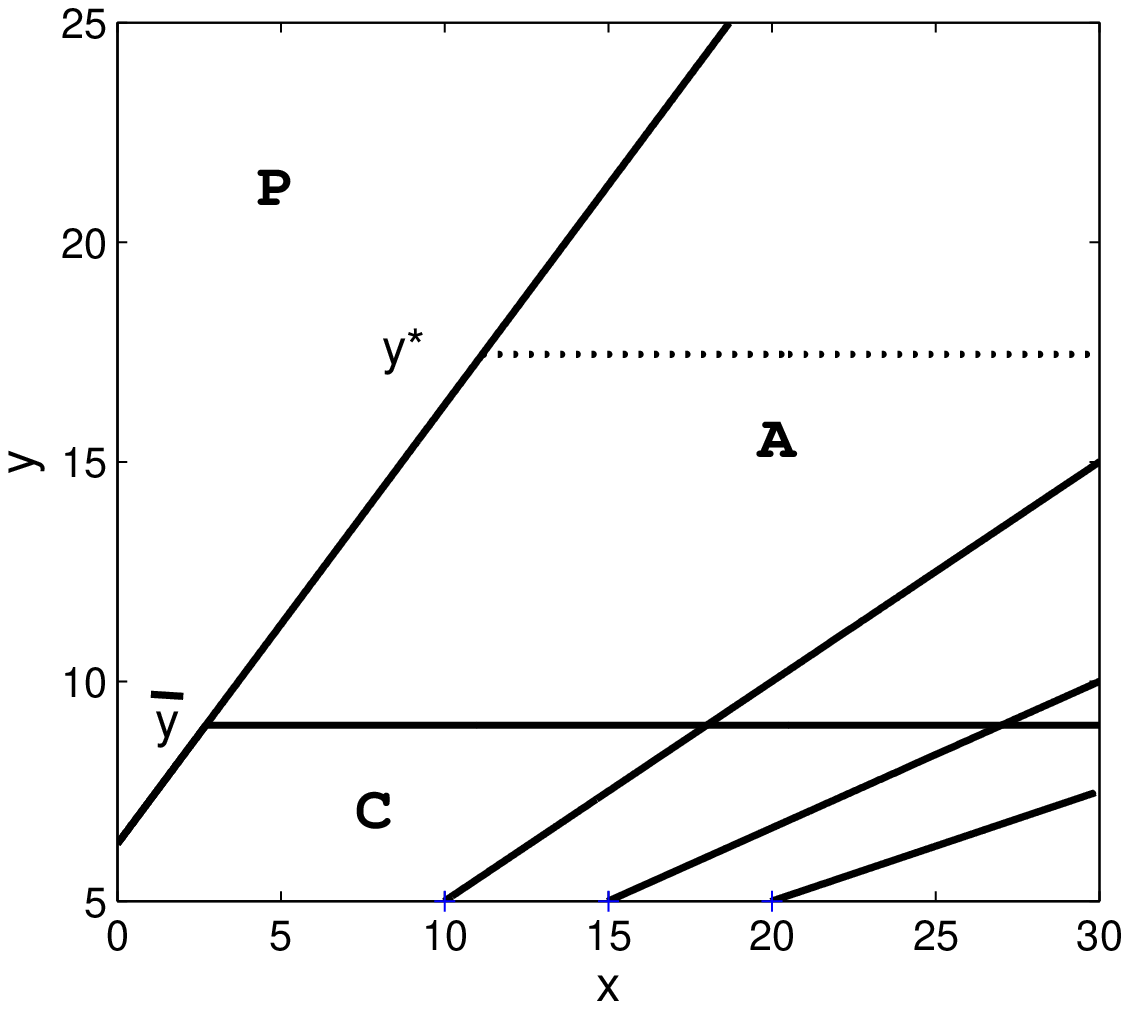}}\qquad
  \subfloat[Attrition rate]{\includegraphics[height=6cm,width=5cm]{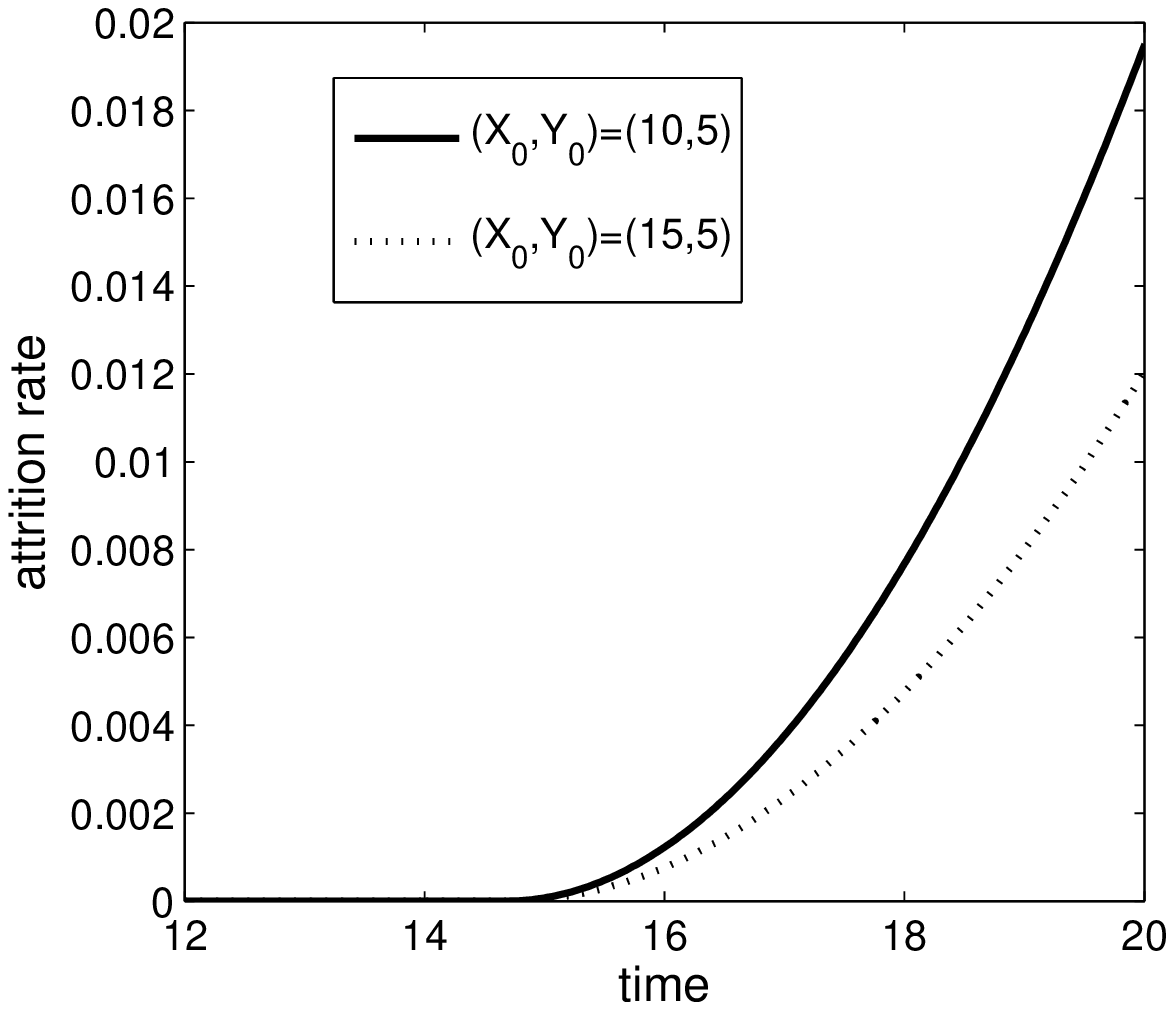}}
  \caption{Paths and attrition rates for deterministic case.}
  \label{fig:dtm}
  % Created with coffeePLdeterministic.m
\end{figure}

If, however, we take the same example, but with $\mu_Y=.06$ and $\mu_X=.02$, then the situation is very different. Figure~\ref{fig:dtm2} shows three sample paths that all end in preemption. Consequently, the attrition rate increases smoothly to 1 at the time the preemption region is hit. This is, again, very different from the typical sample path of the attrition rate as discussed in Section~\ref{sec:example}.

\begin{figure}[htb]
  \centering
  \subfloat[Sample path]{\includegraphics[height=6cm,width=5cm]{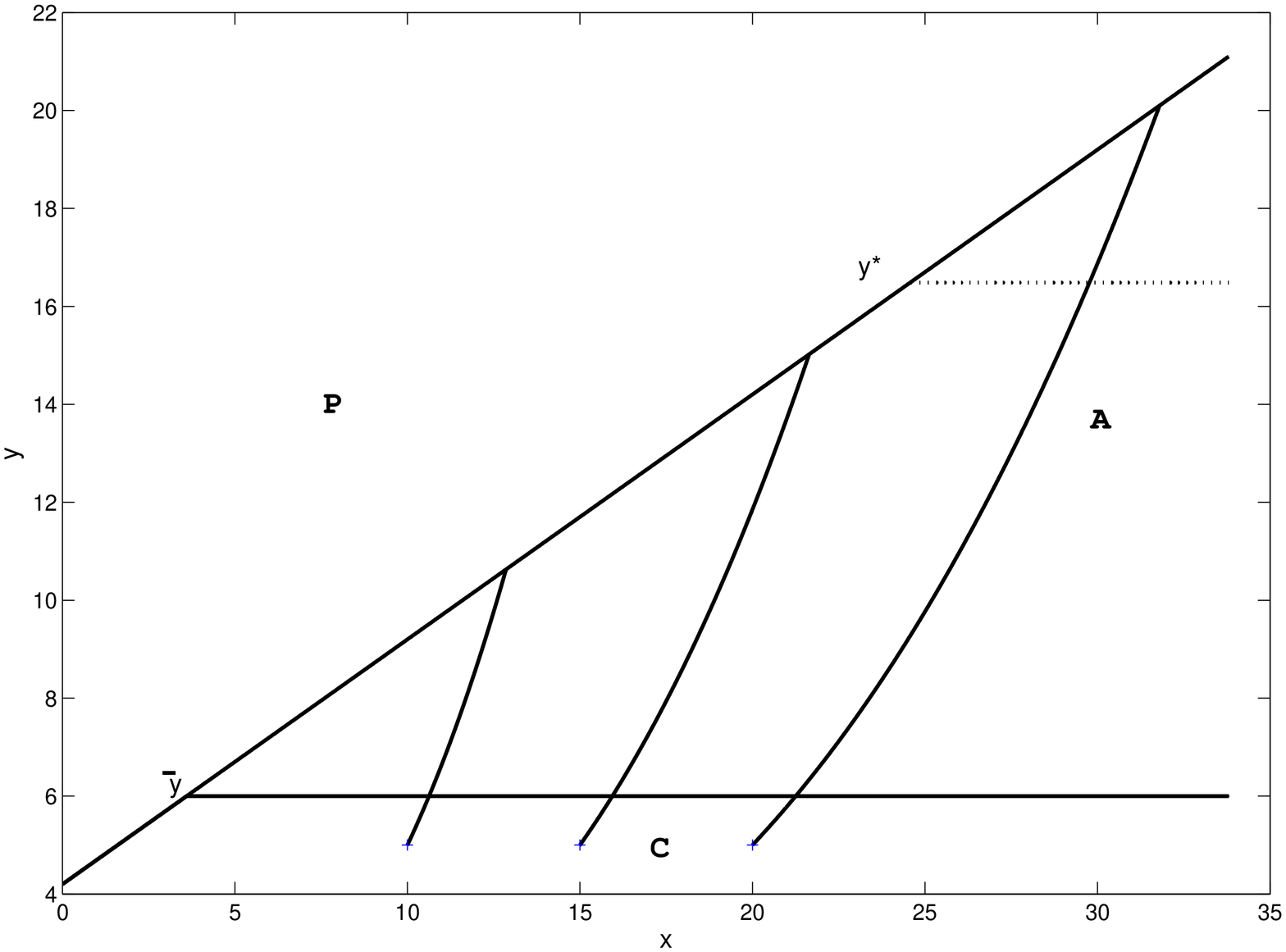}}\qquad
  \subfloat[Attrition rate]{\includegraphics[height=6cm,width=5cm]{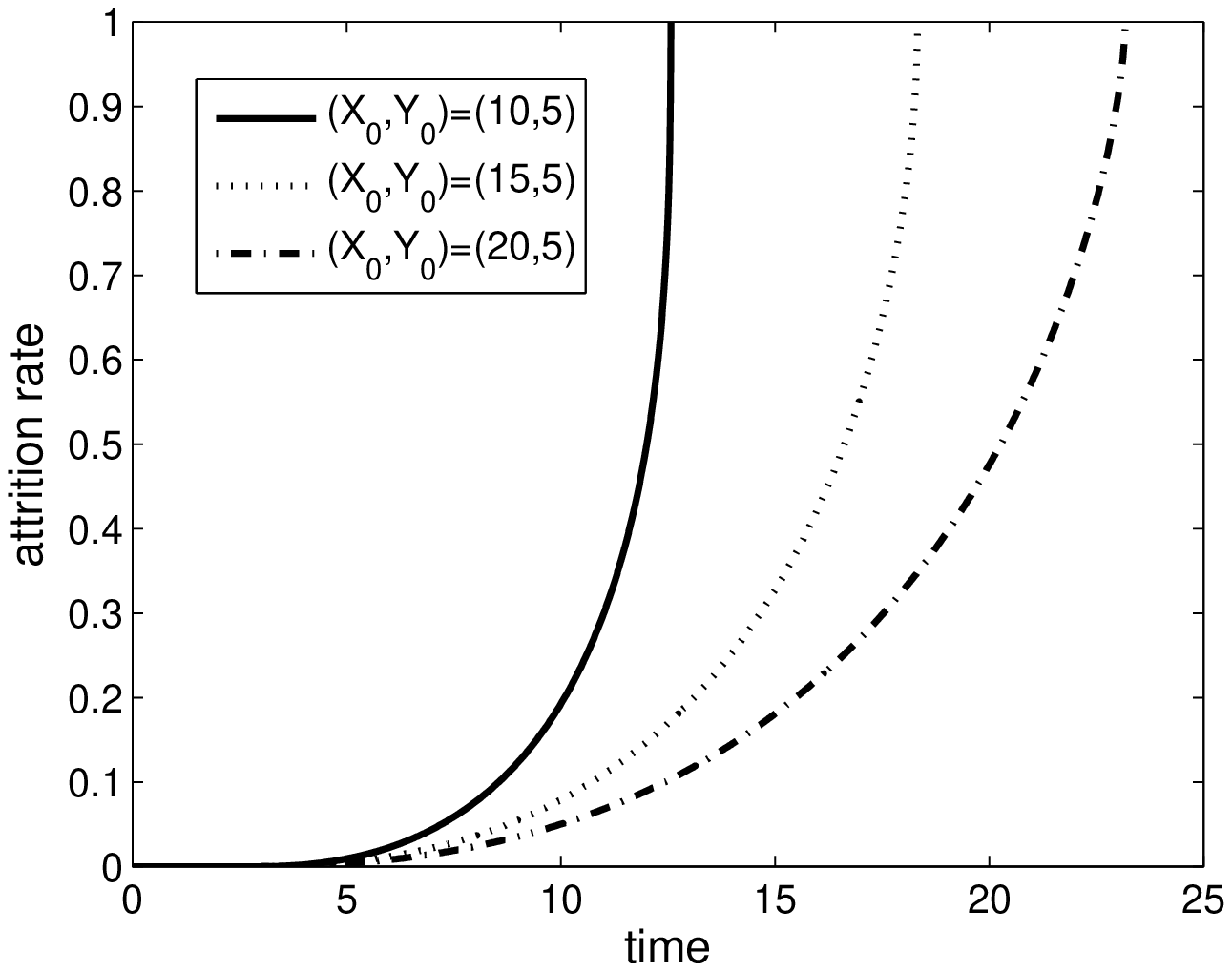}}
  \caption{Paths and attrition rates for deterministic case.}
  \label{fig:dtm2}
  % Created with coffeePLdtm.m
\end{figure}

\section{Concluding Remarks}\label{sec:concl}

In this paper we have built a continuous time stochastic model capturing a natural strategic investment problem that induces a sequence of local first- and second-mover advantages depending on the random evolution of the environment. Our focus has been mainly on regimes with second-mover advantages, as these are far more rarely considered in the real options literature. Despite the non-predictable development of the game we have constructed subgame perfect (Markovian) equilibria, for which we made use of different approaches from the theory of optimal stopping~-- given that we identified stopping problems at the heart of equilibrium construction that are not standard in the literature.

Even though we have made some strong assumptions on the underlying stochastic processes, choosing a quite specific form, it should be clear that our results~-- in particular the economics driving our equilibria~-- can be generalized to other underlying payoff processes. The assumption of geometric Brownian motion is mainly made so that the leader and follower payoff processes can be explicitly computed.

An economically important avenue for future research is to enrich the economic environment by introducing a continuation value for the follower. This could be even extended to a model where firms can switch as often as they wish. Such an extension, however, does not fall into the framework of timing games employed here, but will require a richer model of subgames and, particularly, histories (of investment).

\section*{Acknowledgements}
The authors gratefully acknowledge support from the Department of Economics \& Related Studies at the University of York. Jacco Thijssen gratefully acknowledges support from the Center for Mathematical Economics at Bielefeld University, as well as the Center for Interdisciplinary Research (ZiF) at the same university. Helpful comments were received from participants of the ZiF programme ``Robust Finance: Strategic Power, Knightian Uncertainty, and the Foundations of Economic Policy Advice''.

\newpage
\appendix

\section{Appendix}

\subsection{Proofs}\label{app:proofs}

\begin{proof}[{\bf Proof of Proposition \ref{prop:b(x)}}]
First consider $Y_0=0$, which is absorbing. By our assumptions $y_\cP\geq 0$, so $(\R_+,0)\subset\cP^c$. For $Y\equiv 0$ hence $\tau_\cP=\infty$ (so we face an unconstrained, deterministic problem) and thus $(x,0)\in\cC\Leftrightarrow V_{\tilde L}(x,0)>L(0,0)\Leftrightarrow 0<y^*$ by the solution of the unconstrained problem, or directly from $dL_t=e^{-rt}\bar y\,dt$ for $Y_0=0$ and $y^*>0\Leftrightarrow\bar y>0$. In this case $\tau_{\cC^c}=\infty$ is optimal, and $\tau_{\cC^c}=0$ if $y^*\leq 0$.

We now establish the monotonicity of the stopping and continuation regions for $\tilde L$ in the whole state space by strong path comparisons. Therefore denote the solution to \eqref{SDE} for given initial condition $(X_0,Y_0)=(x,y)\in\R_+^2$ by $(X^x,Y^y)=(xX^1,yY^1)$. Further write the continuation region as $\cC=\{(x,y)\in\R_+^2\mid V_{\tilde L}(x,y)>\tilde L(x,y)\}\subset\cP^c$, where $\tilde L(x,y):=L(0,y)\indi{(x,y)\in\cP^c}+F(0,x)\indi{(x,y)\in\cP}$; cf.\ \eqref{L(y)}. This means that if we fix any $(x_0,y_0)\in\cC$, then there exists a stopping time $\tau^*\in(0,\tau_\cP]$ a.s.\ with $V_{\tilde L}(x_0,y_0)\geq E\bigl[\tilde L_{\tau^*}\bigr]=E\bigl[L(\tau^*,Y^{y_0}_{\tau^*})\bigr]>\tilde L(x_0,y_0)=L(0,y_0)$. Now fix an arbitrary $\varepsilon\in(0,y_0)$ implying that $(X^{x_0},Y^{y_0-\varepsilon})=(X^{x_0},Y^{y_0}-\varepsilon Y^1)$ starts at $(x_0,y_0-\varepsilon)\in\cP^c$ and $\tau^*\leq\tau_\cP\leq\inf\{t\geq 0\mid(X^{x_0}_t,Y^{y_0-\varepsilon}_t)\in\cP\}$.

Hence, $V_{\tilde L}(x_0,y_0-\varepsilon)\geq E\bigl[L(\tau^*,Y^{y_0-\varepsilon}_{\tau^*})\bigr]=E\bigl[L(\tau^*,Y^{y_0}_{\tau^*})-e^{-r\tau^*}\varepsilon Y^1_{\tau^*}/(r-\mu_Y)\bigr]>L(0,y_0)-E\bigl[e^{-r\tau^*}\varepsilon Y^1_{\tau^*}/(r-\mu_Y)\bigr]\geq L(0,y_0)-\varepsilon/(r-\mu_Y)=L(0,y_0-\varepsilon)=\tilde L(x_0,y_0-\varepsilon)$. The last inequality is due to $(e^{-rt}Y^1)$ being a supermartingale by $r>\mu_Y$.

As $\varepsilon$ was arbitrary we can define $b(x):=\sup\{y\geq 0\mid V_{\tilde L}(x,y)>\tilde L(x,y)\}$ for any $x\in\R_+$ where that section of the continuation region is nonempty and conclude $V_{\tilde L}(x,y)>\tilde L(x,y)$ for all $y\in(0,b(x))$. Then we have $b(x)\leq y^*$ because immediate stopping is optimal in the unconstrained problem for $Y_0\geq y^*$, formally $L(0,y)=\tilde L(x,y)\leq V_{\tilde L}(x,y)\leq U_L(0)=L(0,y)$ for any $(X_0,Y_0)=(x,y)\in\cP^c$ with $y\geq y^*$.\footnote{%
Here $U_L$ is the Snell envelope of $L$ , i.e., the value process of the unconstrained stopping problem. See also the corresponding discussion for $\tilde L$ following Proposition~\ref{prop:b(x)}.
}
The same argument shows that the continuation region is empty if $y^*\leq 0$. In this case we can set $b(x):=y^*$.

The section of the continuation region for arbitrary $x\geq 0$ is in fact \emph{only} empty if $y^*\leq 0$, which completes the definition of $b$. Indeed, we have seen at the beginning of the proof that $(\R_+,0)\subset\cC$ if $y^*>0$ (and therefore actually $(x,y)\in\cC$ for all $y\in[0,b(x))$). We could also have applied the following important estimate for $x>0$. If $y^*>0$, i.e., if $\bar y>0$, then $b(x)\geq\min(\bar y,y_\cP+x(r-\mu_Y)/(r-\mu_X))$, which follows from the fact that $L$ is a continuous semimartingale with finite variation part
\begin{equation*}
\int_0^t-e^{-rs}(Y_s-\bar y)\,ds,
\end{equation*}
as can be seen from applying It\=o's Lemma. This drift is strictly positive on $\{Y<\bar y\}$, where stopping $L$ is therefore suboptimal and so it is too for $\tilde L$ (up to $\tau_\cP$).

For the monotonicity of $b$ pick any $(X_0,Y_0)=(x_0,y_0)\in\cC$ and $\tau^*\leq\tau_\cP$ as before and fix an arbitrary $\varepsilon>0$. Then $(X^{x_0+\varepsilon},Y^{y_0})=(X^{x_0}+\varepsilon X^1,Y^{y_0})$ starts at $(x_0+\varepsilon,y_0)\in\cP^c$ and $\tau^*\leq\tau_\cP\leq\inf\{t\geq 0\mid(X^{x_0+\varepsilon}_t,Y^{y_0}_t)\in\cP\}$. Now $V_{\tilde L}(x_0+\varepsilon,y_0)\geq E\bigl[L(\tau^*,Y^{y_0}_{\tau^*})\bigr]>L(0,y_0)=\tilde L(x_0+\varepsilon,y_0)$, whence $b(x_0+\varepsilon)\geq b(x_0)$.

The continuity of $b$ in $x_0=0$ holds by the following estimates. Below we show that $b(x)\geq\min(y_\cP,y^*)$, which together with monotonicity yields a right-hand limit and $b(0),b(0+)\geq\min(y_\cP,y^*)$. On the other hand, by definition $b(x)\leq y_\cP+x(r-\mu_Y)/(r-\mu_X)$ and $b(x)\leq y^*$ as shown above. Hence $b(0)=b(0+)=\min(y_\cP,y^*)$. To verify the continuity of $b$ in $x_0>0$ we show that if $(x_0,y_0)\in\cC$, then also the line $\{(x,y)\in\R_+^2\mid x\in(0,x_0],y=xy_0/x_0\}\subset\cC$.\footnote{%
A graphical illustration helps to convey the continuity argument for $b(x)$.

\centering
   \begin{tikzpicture}[inner sep=0pt,minimum size=0pt,label distance=3pt]
    \draw[->] (-0.1,0) -- (2,0) {}; % x axis
    \draw[->] (0,-0.1) -- (0,2) {}; % y axis
    \draw[-] (0.2,0.5) .. controls (0.4,0.7) and (0.5,0.8) .. (0.75,1) [] {};
    \draw[-] (0.75,1.3) .. controls (1,1.6) and (1,1.6) .. (1.4,1.8) [] {};
    \draw[dashed] (0,0) -- (0.8,1.2) [] {};
    \fill[black] (0.8,1.2) circle (.04);
    \filldraw[fill=white,draw=black] (0.75,1) circle (.04);
    \filldraw[fill=white,draw=black] (0.75,1.3) circle (.04);
    \node at (0,1.8) [label=left:$y$] {};
    \node at (0,1.2) [label=left:$y_0$] {};
    \node at (0.8,0) [label=below:$x_0$] {};
    \node at (1.3,1.7) [label=right:$b(x)$] {};
    \node at (1.3,0.5) [label=above:$\cC$] {};
    \node at (0.35,1) [label=above:$\cC^c$] {};
    \node at (1.9,0) [label=below:$x$] {};
  \end{tikzpicture}
}
As the first step, we check that given $(X_0,Y_0)=(x_0,y_0)\in\cP^c$, $\tau_\cP\leq\inf\{t\geq 0\mid(X^x_t,Y^y_t)\in\cP\}$ for any point $(x,y)$ on the specified line. Indeed, for any $(X^x,Y^y)$ we can write
\begin{align*}
L_t\leq F_t&\Leftrightarrow \frac{Y^y_t}{r-\mu_Y}-\frac{X^x_t}{r-\mu_X}\leq\frac{c_B-c_A}{r}+I=\frac{y_\cP}{r-\mu_Y}\\
&\Leftrightarrow \frac{y}{r-\mu_Y}Y^1_t-\frac{x}{r-\mu_X}X^1_t\leq\frac{y_\cP}{r-\mu_Y}.
\end{align*}
If $y=xy_0/x_0$ this becomes
\begin{align*}
\frac{y}{r-\mu_Y}Y^1_t-\frac{x}{r-\mu_X}X^1_t=\frac{x}{x_0}\biggl(\frac{y_0}{r-\mu_Y}Y^1_t-\frac{x_0}{r-\mu_X}X^1_t\biggr)\leq\frac{y_\cP}{r-\mu_Y}.
\end{align*}
As $y_\cP/(r-\mu_Y)$ is nonnegative, the condition for $(x,y)$ is implied by the one for $(x_0,y_0)$ if $x\in[0,x_0]$. Therefore, if we fix any point on the line $\{(x,y)\in\R_+^2\mid x\in(0,x_0],y=xy_0/x_0\}$, we have $\tau^*\leq\inf\{t\geq 0\mid(X^x_t,Y^y_t)\in\cP\}$, where $\tau^*$ as before satisfies $E\bigl[L(\tau^*,Y^{y_0}_{\tau^*})\bigr]>L(0,y_0)$ for the given $(x_0,y_0)\in\cC$.

In particular $(x,y)\in\cP^c$. Now suppose also $(x,y)\in\cC^c$. Then $L(0,y)=V_{\tilde L}(x,y)\geq E\bigl[L(\tau^*,Y^y_{\tau^*})\bigr]=E\bigl[L(\tau^*,Y^{y_0}_{\tau^*})-e^{-r\tau^*}(y_0-y)/(r-\mu_Y)Y^1_{\tau^*}\bigr]>L(0,y_0)-E\bigl[e^{-r\tau^*}(y_0-y)/(r-\mu_Y)Y^1_{\tau^*}\bigr]\geq L(0,y_0)-(y_0-y)/(r-\mu_Y)=L(0,y)$, a contradiction (where we again used the fact that $(e^{-rt}Y^1)$ is a supermartingale). Thus, $(x,y)\in\cC$.

Now we argue that $\partial\cC\subset\cC^c$, except possibly for the origin. Suppose $y^*>0$, i.e., $\bar y>0$ (otherwise $\cC=\emptyset$). By $b(x)\geq\min(\bar y,y_\cP+x(r-\mu_Y)/(r-\mu_X))$, the only point in $\partial\cC$ with $y=b(x)=0$ can be the origin (and that only if $y_\cP=0$). If $y=Y_0>0$ and $(X_0,Y_0)\in\partial\cC$, then $(X,Y)$ will enter the interior of $\cC^c$ immediately with probability 1. Indeed we have shown that $b(x+h)\leq b(x)+hb(x)/x$ for any $x,h>0$, which implies with the monotonicity of $b$ that if $Y_0=b(X_0)>0$ and $X_0>0$, then $\{Y_t>b(X_t)\}\supset\{Y_t>Y_0\}\cap\{Y_t>Y_0X_t/X_0\}=\{(\mu_Y-\sigma_Y^2/2)t+\sigma_YB^Y_t>0\}\cap\{(\mu_Y-\sigma_Y^2/2)t+\sigma_YB^Y_t-(\mu_X-\sigma_X^2/2)t-\sigma_XB^X_t>0\}$. After normalization we see that the latter two sets are both of the form $\{B_t+\mu t>0\}$ for some Brownian motion $B$ with drift $\mu\in\R$ (we normalize by $\sigma_Y$ in the first and $(\sigma_Y^2-2\rho\sigma_X\sigma_Y+\sigma_X^2)^{1/2}$ in the second case; recall $\sigma_YB^Y$ and $\sigma_YB^Y-\sigma_XB^X$ are not degenerate). Therefore, the hitting times of either set, which we denote by $\tau_1$ and $\tau_2$ are 0 a.s., and $P[\inf\{t\geq 0\mid Y_t>b(X_t)\}=0]\geq 1-P[\tau_1>0]-P[\tau_2>0]=1$, i.e., the correlation $\rho$ between $B^X$ and $B^Y$ is not important. A simpler version of the argument applies to the case $X_0=0$ and $Y_0=b(0)>0$, when $\{Y_t>b(X_t)\}=\{Y_t>Y_0\}$.

It remains to prove that $b(x)\geq\min(y_\cP,y^*)$. If $y^*\leq y_\cP$ then $\tau_\cP\geq\tau^*(y):=\inf\{t\geq 0\mid Y_t\geq y^*\}$ for $(X_0,Y_0)=(x,y)\in\cP^c$ and thus $E\bigl[L_{\tau^*(y)}\bigr]=U_L(0)\geq V_{\tilde L}(x,y)\geq E\bigl[L_{\tau^*(y)}\bigr]$. Hence, $U_L(0)=V_{\tilde L}(x,y)>L(0,y)=\tilde L(x,y)$ for $y<y^*$, implying $b(x)\geq y^*$.

For $y^*>y_\cP$ we will show that $\tilde L(x,y)=L(0,y)<E\bigl[L_{\tau(y_\cP)}\bigr]\leq V_{\tilde L}(x,y)$ for any $(x,y)$ with $y<y_\cP$, where $\tau(y_\cP):=\inf\{t\geq 0\mid Y_t\geq y_\cP\}\leq\tau_\cP$, which implies that $b(x)\geq y_\cP$. On $\{Y_0=y<y_\cP\}$ we obtain by the definition of $y_\cP$ that $E\bigl[L_{\tau(y_\cP)}\bigr]=-c_0/r+(y/y_\cP)^{\beta_1}(c_0-c_A)/r$,\footnote{\label{fn:U_L}%
For any threshold $\hat y>0$ and its hitting time $\tau(\hat y):=\inf\{t\geq 0\mid Y_t\geq \hat y\}$ by the geometric Brownian motion $Y$ we have by standard results
\begin{equation*}
E[L_{\tau(\hat y)}]=-\frac{c_0}{r}+\biggl(\frac{Y_0}{\hat y}\biggr)^{\beta_1}\biggl(\frac{\hat y}{r-\mu_Y}-\frac{c_B-c_0}{r}-I\biggr)\text{ if }Y_0<\hat y.
\end{equation*}
}
which is a continuous function of $y$ that extends to $E\bigl[L_{\tau(y_\cP)}\bigr]=L(0,y_\cP)$ for $Y_0=y=y_\cP$ and to $E\bigl[L_{\tau(y_\cP)}\bigr]=-c_0/r$ for $Y_0=y=0$. Hence, $E\bigl[L_{\tau(y_\cP)}\bigr]-L(0,y)$ vanishes at $y=y_\cP$, while at $y=0$ it attains $I+(c_B-c_0)/r$, which is strictly positive if $y^*>0$. Now $\partial_y E\bigl[L_{\tau(y_\cP)}-L(0,y)\bigr]<0$ for all $y\in(0,y_\cP)$ iff $\beta_1(c_0-c_A)(y/y_\cP)^{\beta_1-1}<c_B-c_A+rI$ for all $y\in(0,y_\cP)$, iff $\beta_1(c_0-c_A)\leq c_B-c_A+rI$ (since $\beta_1>1$), i.e., iff $y^*\geq y_\cP$.
\end{proof}

\begin{proof}[{\bf Proof of Proposition \ref{prop:dD_Lswitch}}]
As in the proof of Proposition \ref{prop:b(x)} we begin with the deterministic case $Y_0=0$, whence $\tau_\cP=\infty$ and $dL_t=e^{-rt}\bar y\,dt$. In this case the Snell envelope of $\tilde L=L$ is the smallest nonincreasing function dominating $L$. For $\bar y\geq 0\Leftrightarrow y^*\geq 0$, it is just the constant $L_\infty=-c_0/r$ and thus $dD_{\tilde L}=0$, while for $\bar y<0\Leftrightarrow y^*<0$ it is $L$ itself and thus $dD_{\tilde L}=-dL$. Both correspond to the claim, cf.\ also Remark \ref{rem:Y=0}.

For $Y_0>0=X_0$, the solution is by Proposition \ref{prop:b(x)} to stop the first time $Y$ exceeds $b(0)=\min(y_\cP,y^*)$. Then the Snell envelope is explicitly known for $t\leq\tau_\cP$, as in footnote \ref{fn:U_L} with $\hat y=\min(y_\cP,y^*)$, inserting $Y_t$ and discounting the brackets by $e^{-rt}$ (and $U_{\tilde L}=\tilde L$ for $Y\geq\hat y$). In this case $dD_{\tilde L}$ can be determined as the drift of $U_{\tilde L}$ from It\=o's formula. One can also reduce the following argument to the 1-dimensional case.

For $X_0,Y_0>0$ first note that the Snell envelope $U_{\tilde L}=V_{\tilde L}(X,Y)$ is continuous because $V_{\tilde L}(\cdot)$ is a continuous function; see, e.g., \cite{Krylov80}. The question when $D_{\tilde L}$ is exactly the decreasing component of the continuous semimartingale $\tilde L$ in the stopping region is addressed by \citet{Jacka93}: It is the case if the local time of the nonnegative semimartingale $U_{\tilde L}-\tilde L$ at 0 is trivial, i.e., if $L^0(U_{\tilde L}-\tilde L)\equiv 0$ (a.s.). We can show that the latter is indeed true in our model by applying the argument of Jacka's Theorem 6.

Specifically, It\=o's Lemma shows that $\tilde L$ is a continuous semimartingale with finite variation part $A_t:=\int_0^t-\indi{s<\tau_\cP}e^{-rs}(Y_s-c_B+c_0-rI)\,ds$. Denote its decreasing part by $A^-$, which here satisfies
\begin{equation*}
dA^-_t=\indi{Y>\bar y}\,e^{-rt}(Y_t-c_B+c_0-rI)\,dt
\end{equation*}
for $t<\tau_\cP$. By Theorem 3 of \cite{Jacka93}, $dL^0_t(U_{\tilde L}-\tilde L)\leq \indi{U_{\tilde L}=\tilde L,t<\tau_\cP}2\,dA^-_t$, and as in his Theorem 6, $dL^0(U_{\tilde L}-\tilde L)$ is supported by $\{(X,Y)\in\partial\cC\}$.

If we let $A^-$ as above for all $t\in\R_+$,
\begin{equation*}
E\bigl[L^0_t(U_{\tilde L}-\tilde L)\bigr]\leq E\biggl[2\int_0^t\indi{(X,Y)\in\partial\cC,s<\tau_\cP}\,dA^-_s\biggr]\leq E\biggl[2\int_0^t\indi{(X,Y)\in\partial\cC}\,dA^-_s\biggr].
\end{equation*}
$dA^-$ apparently has a Markovian density with respect to Lebesgue measure on $R_+$, and our underlying diffusion $(X,Y)$ has a log-normal transition distribution, which thus has a density with respect to Lebesgue measure on $R_+^2$. As $\partial\cC=\{(x,y)\in\R_+^2\mid y=b(x)\}$ has Lebesgue measure 0 in $R_+^2$, we conclude like in the proof of Theorem 6 of \cite{Jacka93} that $L^0(U_{\tilde L}-\tilde L)\equiv 0$ (a.s.).
\end{proof}

\begin{proof}[{\bf Proof of Proposition \ref{prop:DG>0}}]
We need to show that
\begin{equation}\label{eq:eqlF>L_a}
\int_0^{\tau_\cP}\frac{dG^\vartheta_i(t)}{1-G^\vartheta_i(t)}=\int_0^{\tau_\cP}\frac{\indi{Y_t\geq b(X_t)}(Y_t-\bar y)\,dt}{X_t/(r-\mu_X)-(Y_t-y_\cP)/(r-\mu_Y)}
<\infty
\end{equation}
on $\{\tau_\cP<\infty\}$ a.s. Note that $\tau_\cP=\inf\{t\geq\vartheta\mid L_t\geq F_t\}$ if $(X_0,Y_0)\not=(0,0)$ or $y_\cP>0$ (see the proof of Proposition \ref{prop:eqlF>L}). As $Y$ is continuous, it is bounded on $[0,\tau_\cP]$ where this interval is finite. Hence we may just use $dt$ as the numerator in \eqref{eq:eqlF>L_a}. We will also ignore $(r-\mu_X)$ and $(r-\mu_Y)$, which may be encoded in the initial position $(X_0,Y_0)$. By the strong Markov property we set $\vartheta=0$. Hence, we will prove that
\begin{equation}\label{eq:int1/Z}
\int_0^{\tau_0}\frac{dt}{X_t-Y_t+a}
<\infty
\end{equation}
on $\{\tau_0<\infty\}$ a.s.\ for nonnegative geometric Brownian motions $X$ and $Y$ (under our nondegeneracy assumption $\sigma_X^2,\sigma_Y^2>0$, $\abs{\rho}<1$), some fixed level $a>0$ and the stopping time $\tau_\varepsilon:=\inf\{t\geq 0\mid X_t-Y_t+a\leq\varepsilon\}$ with $\varepsilon=0$. We will treat the special case $y_\cP=a=0$ at the very end of the proof. Define the process $Z:=X-Y+a$ to simplify notation.

As a first step and tool, we derive the weaker result $E[\int_0^{\tau_0\wedge T}\ln(Z_t)\,dt]\in\R$ for any time $T>0$ (which implies also $\int_0^{\tau_0}\abs{\ln(Z_t)}\,dt<\infty$ on $\{\tau_0<\infty\}$ a.s.\ by the arguments towards the end of the proof).

Define the function $f\colon\R_+\to\R$ by
\begin{equation*}
f(x)=x\ln(x)-x\in[-1,x^2]
\end{equation*}
and the function $F\colon\R_+\to\R$ by $F(x)=\int_0^xf(y)\,dy\in[-x,x^3]$, such that $F''(x)=\ln(x)$ for all $x>0$. For localization purposes fix an $\varepsilon>0$ and a time $T>0$. By It\=o's formula we have
\begin{align*}
F(Z_{\tau_\varepsilon\wedge T})=F(Z_0)+&\int_0^{\tau_\varepsilon\wedge T}f(Z_t)\,dZ_t+\frac12\int_0^{\tau_\varepsilon\wedge T}\ln(Z_t)\,d[Z]_t\\
=F(Z_0)+&\int_0^{\tau_\varepsilon\wedge T}f(Z_t)\bigl(\mu_XX_t-\mu_YY_t\bigr)\,dt\\
+&\int_0^{\tau_\varepsilon\wedge T}f(Z_t)\Bigl(\sigma_XX_t\,dB^X_t-\sigma_YY_t\,dB^Y_t\Bigr)\\
+&\frac12\int_0^{\tau_\varepsilon\wedge T}\ln(Z_t)\bigl(\underbrace{\sigma^2_XX^2_t+\sigma^2_YY^2_t-2\rho\sigma_X\sigma_YX_tY_t}_{=:\sigma^2(X_t,Y_t)}\bigr)\,dt.
\end{align*}
We want to establish $\lim_{\varepsilon\searrow 0}E[F(Z_{\tau_\varepsilon\wedge T})]$, first in terms of the integrals, for which we need some estimates.

In order to eliminate the second, stochastic integral by taking expectations, it is sufficient to verify that $\indi{t<\tau_0}f(Z_t)X_t$ and $\indi{t<\tau_0}f(Z_t)Y_t$ are square \nbd{P\otimes dt}integrable on $\Omega\times[0,T]$. We have $\abs{f(Z_t)}\leq 1+Z^2_t$. For $t\leq\tau_0$ further $0\leq Z_t\leq X_t+a$ by $Y_t\geq 0$ and hence $Z^2_t\leq(X_t+a)^2\leq 2X^2_t+2a^2$. The sought square-integrability with $X$ follows now from the \nbd{P\otimes dt}integrability of $X^n_t$ on $\Omega\times[0,T]$ for any $n\in\N$ and analogously that for $Y$ thanks to $0\leq Y_t\leq X_t+a$ for $t\leq\tau_0$.

The same estimates guarantee that the expectation of the first integral converges to the finite expectation at $\tau_0$ as $\varepsilon\searrow 0$. For the third integral we have $\ln(Z_t)\leq Z_t\leq X_t+a$. The second term $\sigma^2(X_t,Y_t)$ is bounded from below by $(\sigma_XX_t-\sigma_YY_t)^2\geq 0$ and from above by $(\sigma_XX_t+\sigma_YY_t)^2\leq((\sigma_X+\sigma_Y)X_t+\sigma_Ya)^2\leq 2(\sigma_X+\sigma_Y)^2X^2_t+2\sigma_Y^2a^2$ for $t\leq\tau_0$ (supposing $\sigma_X,\sigma_Y>0$ wlog.). Hence the positive part of the integrand is bounded by a \nbd{P\otimes dt}integrable process on $\Omega\times[0,T]$, while the negative part converges monotonically, and we may take the limit of the expectation of the whole integral.

That the latter is finite follows from analysing the limit of the LHS, $\lim_{\varepsilon\searrow 0}E[F(Z_{\tau_\varepsilon\wedge T})]$, directly. We have $F(Z_{\tau_\varepsilon\wedge T})=\indi{\tau_\varepsilon\leq T}F(\varepsilon)+\indi{T<\tau_\varepsilon}F(Z_T)$, and it is continuous in $\varepsilon\searrow 0$. For $T<\tau_0$ again $0\leq Z_T\leq X_T+a$ and thus $\abs{F(Z_T)}\leq\abs{Z_T}+\abs{Z_T}^3\leq (X_T+a)+(X_T+a)^3$. As also $\abs{F(\varepsilon)}\leq 1$ for all $\varepsilon\leq 1$, $\abs{F(Z_{\tau_\varepsilon\wedge T})}$ is bounded by an integrable random variable as $\varepsilon\searrow 0$. Consequently, $\lim_{\varepsilon\searrow 0}E[F(Z_{\tau_\varepsilon\wedge T})]=E[F(Z_{\tau_0\wedge T})]\in\R$ and also $E[\int_0^{\tau_0\wedge T}\ln(Z_t)\sigma^2(X_t,Y_t)\,dt]\in\R$ on the RHS. In the integral we may ignore the term $\sigma^2(X_t,Y_t)$, which completes the first step. Indeed, with $\abs{\rho}<1$ we can only have $\sigma^2(X_t,Y_t)=0$ if $X_t=Y_t=0$, i.e., if $Z_t=a$. Therefore $\inf\{\sigma^2(X_t,Y_t)\mid Z_t\leq\varepsilon\}>0$ for any fixed $\varepsilon\in(0,a)$, so $\sigma^2(X_t,Y_t)$ does not ``kill'' the downside of $\ln(Z_t)$ and
\begin{equation*}
E\biggl[\int_0^{\tau_0\wedge T}\ln(Z_t)\,dt\biggr]\in\R.
\end{equation*}

In the following we further need $E[\int_0^{\tau_0\wedge T}\ln(Z_t)X_t\,dt]\in\R$, which obtains as follows. With $\abs{\rho}<1$, $\inf\{\sigma^2(X_t,Y_t)\mid Z_t\leq\varepsilon\}$ is attained only with the constraint binding, $Y_t=X_t+a-\varepsilon$. Thus, for $Z_t\leq\varepsilon$, $\sigma^2(X_t,Y_t)\geq\sigma^2(X_t,X_t+a-\varepsilon)$. The latter is a quadratic function of $X_t$, with $X^2_t$ having coefficient $\sigma^2(1,1)>0$ given $\abs{\rho}<1$. The quadratic function hence exceeds $X$ for all $X$ sufficiently large, i.e., we can pick $K>0$ such that $\sigma^2(X_t,Y_t)\geq X_t$ on $\{Z_t\leq\varepsilon\}\cap\{X_t\geq K\}$. $X_t$ thus does not ``blow up'' the downside of $\ln(Z_t)$ more than $\sigma^2(X_t,Y_t)$.

We are now ready to analyse
\begin{align*}
f(Z_{\tau_\varepsilon\wedge T})=f(Z_0)+&\int_0^{\tau_\varepsilon\wedge T}\ln(Z_t)\,dZ_t+\frac12\int_0^{\tau_\varepsilon\wedge T}\frac{1}{Z_t}\,d[Z]_t\\
=f(Z_0)+&\int_0^{\tau_\varepsilon\wedge T}\ln(Z_t)\bigl(\mu_XX_t-\mu_YY_t\bigr)\,dt\\
+&\int_0^{\tau_\varepsilon\wedge T}\ln(Z_t)\Bigl(\sigma_XX_t\,dB^X_t-\sigma_YY_t\,dB^Y_t\Bigr)\\
+&\frac12\int_0^{\tau_\varepsilon\wedge T}\frac{1}{Z_t}\sigma^2(X_t,Y_t)\,dt
\end{align*}
when taking the limit $\varepsilon\searrow 0$ under expectations as before. By our final observations of step one the first integral converges (note again $Y_t\leq X_t+a$ for $t\leq\tau_0$). In the second, stochastic integral we now have $\abs{\ln(Z_t)}\leq\abs{\ln(\varepsilon)}+\abs{Z_t}$ for $t\leq\tau_\varepsilon$, an even smaller bound than above, making the expectation vanish. In the third integral, $\sigma^2(X_t,Y_t)/Z_t\geq 0$ for $t\leq\tau_0$, so monotone convergence holds.

On the LHS, $\abs{f(Z_{\tau_\varepsilon\wedge T})}$ is bounded by an integrable random variable for all $\varepsilon\leq 1$ analogously to the first step, implying $\lim_{\varepsilon\searrow 0}E[f(Z_{\tau_\varepsilon\wedge T})]=E[f(Z_{\tau_0\wedge T})]\in\R$ and thus $E[\int_0^{\tau_0\wedge T}\sigma^2(X_t,Y_t)/Z_t\,dt]<\infty$. By the same arguments brought forward at the end of the first step we can again ignore $\sigma^2(X_t,Y_t)$.

With $E[\int_0^{\tau_0\wedge T}1/Z_t\,dt]<\infty$, $P[\{\tau_0\leq T\}\cap\{\int_0^{\tau_0}1/Z_t\,dt=\infty\}]=0$. As $T$ was arbitrary, we may take the union over all integer $T$ to conclude
$P[\{\tau_0<\infty\}\cap\{\int_0^{\tau_0}1/Z_t\,dt=\infty\}]=0$. This completes the proof of \eqref{eq:int1/Z} for the case $a>0$ (and that of \eqref{eq:eqlF>L_a} for $y_\cP>0$).

Some modification is due for the case $y_\cP=a=0$ and if $(X_0,Y_0)\not=(0,0)$ (otherwise $\tau_0=0$ and \eqref{eq:int1/Z} is trivial). Then $\sigma^2(X_t,Y_t)$ is not bounded away from 0 for $Z_t$ small, it may be an important factor in the integrability of $\ln(Z_t)\sigma^2(X_t,Y_t)$ when $(X_t,Y_t)$ is close to the origin. In order to infer the required well-behaved limit of $E[\int_0^{\tau_\varepsilon\wedge T}\ln(Z_t)\bigl(\mu_XX_t-\mu_YY_t\bigr)\,dt]$ and to finally remove $\sigma^2(X_t,Y_t)$ from $\int_0^{\tau_0}\sigma^2(X_t,Y_t)/Z_t\,dt$, one can employ another localization procedure: Fix a small $\delta>0$ and use the minimum of $\sigma_\delta:=\inf\{t\geq 0\mid X_t+Y_t<\delta\}$ and $\tau_\varepsilon\wedge T$ everywhere above. Then $\sigma^2(X_t,Y_t)$ is bounded away from 0 on $[0,\sigma_\delta\wedge\tau_0]$ for $\abs{\rho}<1$ and again bounded below by a quadratic function with $X^2_t$ having coefficient $\sigma^2(1,1)>0$. The result now obtains as above for all paths with $X_t+Y_t\geq\delta$ on $[0,\tau_0]$. As $\delta>0$ is indeed arbitrary and any path with $(X_0,Y_0)\not=(0,0)$ is bounded away from the origin on $[0,\tau_0]$ finite by continuity, the claim follows.
\end{proof}

\subsection{Outcome Probabilities}\label{app:outcome}

The following definition is a simplification of that in \cite{RiedelSteg14}, resulting from right-continuity of any $\alpha^\vartheta_i(\cdot)$ also where it takes the value 0.

Define the functions $\mu_L$ and $\mu_M$ from $[0,1]^2\setminus (0,0)$ to $[0,1]$ by
\begin{align*}
\mu_L(x,y):=\frac{x(1-y)}{x+y-xy}\qquad\text{and}\qquad\mu_M(x,y):=\frac{xy}{x+y-xy}.
\end{align*}
$\mu_L(a_i,a_j)$ is the probability that firm $i$ stops first in an infinitely repeated stopping game where $i$ plays constant stage stopping probabilities $a_i$ and firm $j$ plays constant stage probabilities $a_j$. $\mu_M(a_i,a_j)$ is the probability of simultaneous stopping and $1-\mu_L(a_i,a_j)-\mu_M(a_i,a_j)=\mu_L(a_j,a_i)$ that of firm $j$ stopping first.

\begin{definition}\label{def:outcome}
Given $\vartheta\in\T$ and a pair of extended mixed strategies $\bigl(G^\vartheta_1,\alpha^\vartheta_1\bigr)$ and $\bigl(G^\vartheta_2,\alpha^\vartheta_2\bigr)$, the \emph{outcome probabilities} $\lambda^\vartheta_{L,1}$, $\lambda^\vartheta_{L,2}$ and $\lambda^\vartheta_M$ at $\hat\tau^\vartheta:=\inf\{t\geq\vartheta\mid\alpha^\vartheta_1(t)+\alpha^\vartheta_2(t)>0\}$ are defined as follows. Let $i,j\in\{1,2\}$, $i\not=j$.

\noindent
If $\hat\tau^\vartheta<\hat\tau^\vartheta_j:=\inf\{t\geq\vartheta\mid\alpha^\vartheta_j(t)>0\}$, then
\begin{align*}
\lambda^\vartheta_{L,i}:={}&\bigl(1-G_i^\vartheta(\hat\tau^\vartheta-)\bigr)\bigl(1-G_j^\vartheta(\hat\tau^\vartheta)\bigr),\\[6pt]
\lambda^\vartheta_M:={}&\bigl(1-G_i^\vartheta(\hat\tau^\vartheta-)\bigr)\alpha^\vartheta_i(\hat\tau^\vartheta)\Delta G_j^\vartheta(\hat\tau^\vartheta).
\end{align*}
If $\hat\tau^\vartheta<\hat\tau^\vartheta_i:=\inf\{t\geq\vartheta\mid\alpha^\vartheta_i(t)>0\}$, then
\begin{align*}
\lambda^\vartheta_{L,i}:={}&\bigl(1-G_j^\vartheta(\hat\tau^\vartheta-)\bigr)\bigl(1-\alpha_j(\hat\tau^\vartheta)\bigr)\Delta G_i^\vartheta(\hat\tau^\vartheta),\\[6pt]
\lambda^\vartheta_M:={}&\bigl(1-G_j^\vartheta(\hat\tau^\vartheta-)\bigr)\alpha^\vartheta_j(\hat\tau^\vartheta)\Delta G_i^\vartheta(\hat\tau^\vartheta).
\end{align*}
If $\hat\tau^\vartheta=\hat\tau^\vartheta_1=\hat\tau^\vartheta_2$ and $\alpha^\vartheta_1(\hat\tau^\vartheta)+\alpha^\vartheta_2(\hat\tau^\vartheta)>0$, then
\begin{align*}
\lambda^\vartheta_{L,i}:={}&\bigl(1-G_i^\vartheta(\hat\tau^\vartheta-)\bigr)\bigl(1-G_j^\vartheta(\hat\tau^\vartheta-)\bigr)\mu_L(\alpha^\vartheta_i(\hat\tau^\vartheta),\alpha^\vartheta_j(\hat\tau^\vartheta)),\\[6pt]
\lambda^\vartheta_M:={}&\bigl(1-G_i^\vartheta(\hat\tau^\vartheta-)\bigr)\bigl(1-G_j^\vartheta(\hat\tau^\vartheta-)\bigr)\mu_M(\alpha^\vartheta_1(\hat\tau^\vartheta),\alpha^\vartheta_2(\hat\tau^\vartheta)).
\end{align*}
If $\hat\tau^\vartheta=\hat\tau^\vartheta_1=\hat\tau^\vartheta_2$ and $\alpha^\vartheta_1(\hat\tau^\vartheta)+\alpha^\vartheta_2(\hat\tau^\vartheta)=0$, then
\begin{align*}
\lambda^\vartheta_{L,i}:={}&\bigl(1-G_i^\vartheta(\hat\tau^\vartheta-)\bigr)\bigl(1-G_j^\vartheta(\hat\tau^\vartheta-)\bigr)\,\frac{1}{2}\!\begin{aligned}[t]
\biggl\{&\liminf_{\underset{\alpha^\vartheta_i(t)+\alpha^\vartheta_j(t)>0}{t\searrow\hat\tau^\vartheta}}\mu_L(\alpha^\vartheta_i(t),\alpha^\vartheta_j(t))\\
+&\limsup_{\underset{\alpha^\vartheta_i(t)+\alpha^\vartheta_j(t)>0}{t\searrow\hat\tau^\vartheta}}\mu_L(\alpha^\vartheta_i(t),\alpha^\vartheta_j(t))\biggr\},
\end{aligned}\\
\lambda^\vartheta_M:={}&0.
\end{align*}
\end{definition}

\begin{remark}\label{rem:outcome}
\noindent
\begin{enumerate}
\item
$\lambda^\vartheta_M$ is the probability of simultaneous stopping at $\hat\tau^\vartheta$, while $\lambda^\vartheta_{L,i}$ is the probability of firm $i$ becoming the leader, i.e., that of firm $j$ becoming follower. It holds that $\lambda^\vartheta_M+\lambda^\vartheta_{L,i}+\lambda^\vartheta_{L,j}=\bigl(1-G_i^\vartheta(\hat\tau^\vartheta-)\bigr)\bigl(1-G_j^\vartheta(\hat\tau^\vartheta-)\bigr)$. Dividing by $\bigl(1-G_i^\vartheta(\hat\tau^\vartheta-)\bigr)\bigl(1-G_j^\vartheta(\hat\tau^\vartheta-)\bigr)$ where feasible yields the corresponding conditional probabilities.

\item
If any $\alpha^\vartheta_\cdot=1$, then no limit argument is needed. Otherwise both $\alpha^\vartheta_\cdot$ are right-continuous and the corresponding limit of $\mu_M$ exists. $\mu_L$, however, has no continuous extension at the origin, whence we use the symmetric combination of $\liminf$ and $\limsup$, ensuring consistency whenever the limit does exist. If the limit in a potential equilibrium does not exist, both firms will be indifferent about the roles; see Lemma A.5 in \cite{RiedelSteg14}.
\end{enumerate}
\end{remark}

%\bibliography{jhs}

\bibliography{BibCoffee}
\end{document}